\documentclass[copyright,creativecommons,USenglish]{eptcs} 
\providecommand{\event}{GandALF 2013} 

\usepackage{amsmath,amssymb,amsthm}  
\usepackage{eucal}
\usepackage{subfig}
\usepackage{paralist}
\usepackage{xspace}
\usepackage{algorithm2e}
\usepackage{thm-restate} 
\newcommand{\showfigure}[1]{#1}
\usepackage{tikz}
\showfigure{\usetikzlibrary{positioning,automata,calc,shapes}}
\DeclareFontFamily{U}{MnSymbolD}{}
\DeclareFontShape{U}{MnSymbolD}{m}{n}{<-> MnSymbolD10}{}
\DeclareRobustCommand{\MNSlessdot}{\mbox{\usefont{U}{MnSymbolD}{m}{n}\char114}}
\DeclareRobustCommand{\MNSleqdot}{\mbox{\usefont{U}{MnSymbolD}{m}{n}\char116}}
\DeclareRobustCommand{\eq}{=}

\renewcommand{\phi}{\varphi}

\newcommand\hide[1]{}

\newcommand{\standout}[1]{{\medskip \noindent {{\textbf{#1}}}}}

\newcommand{\zug}[1]{{\langle #1  \rangle}}
\newcommand{\rzug}[1]{{\langle #1  \rangle}}
\newcommand{\set}[1]{\text{$\{#1\}$}} 
\newcommand\splus{\ensuremath{\!+\!}}
\newcommand\sminus{\ensuremath{\!-\!}}
\newcommand\stimes{\ensuremath{\!\times\!}}

\newcommand\kahler{K\"ahler\xspace}
\newcommand\konig{K\"onig\xspace}
\newcommand\buchi{B\"uchi\xspace}

\newcommand{\N}{\mbox{I$\!$N}}
\renewcommand{\N}{{\mathbb N}}
\newcommand{\Q}{\mathbb{Q}}
\newcommand{\U}{{\cal U}}
\newcommand{\gl}{\ensuremath{gl}}
\newcommand\A{{\cal A}}
\newcommand{\B}{{\cal B}}
\newcommand{\equal}{\ensuremath{=}}

\newcommand{\Z}{{\mathbb Z}}

\newcommand{\DAG}{\textsc{dag}\xspace}
\newcommand{\DPEW}{\textsc{DPeW}\xspace}
\newcommand{\DREW}{\textsc{DReW}\xspace}
\newcommand{\DAGs}{\textsc{dag}s\xspace}
\newcommand{\pri}{\gamma}

\newcommand{\rhoprec}{\ensuremath{\rho_{{\scriptscriptstyle S,\preceq}}}}

\def\squarebox#1{\hbox to #1{\hfill\vbox to #1{\vfill}}}
\renewcommand{\qed}{\hspace*{\fill}
          \vbox{\hrule\hbox{\vrule\squarebox{.667em}\vrule}\hrule}\smallskip}
\newcommand{\abs}[1]{\lvert#1\rvert}

\renewenvironment{proof}{\begin{trivlist}
\item[\hspace{\labelsep}{\bf\noindent Proof: }]
}{\qed\end{trivlist}}

\theoremstyle{plain}
\newtheorem{theorem}{Theorem}[section]
\newtheorem{lemma}[theorem]{Lemma} 
\newtheorem{corollary}[theorem]{Corollary}
\theoremstyle{definition}
\newtheorem{definition}[theorem]{Definition}
\theoremstyle{remark}

\newtheoremstyle{note}
{1.5\topsep}
{1.5\topsep}
{}
{}
{\itshape}
{.}
{ }
{}
\theoremstyle{note}
\newtheorem{example}[theorem]{Example}

\newcommand{\G}{\ensuremath{{G}}\xspace}
\newcommand{\bQ}{\ensuremath{{\bf Q}}\xspace}
\newcommand{\q}{\ensuremath{{\bf q}}\xspace}
\newcommand{\qp}{\ensuremath{{\bf q'}}\xspace}

\newcommand{\Gprime}{\ensuremath{G'}\xspace}

\newcommand{\TGP}{\ensuremath{T}\xspace}

\newcommand{\lr}[1]{\ensuremath{\,\MNSlessdot_{#1}\,}}
\newcommand{\lreq}[1]{\ensuremath{\,\MNSleqdot_{#1}\,}}

\newcommand{\lrb}{\ensuremath{\,\MNSlessdot\,}}
\newcommand{\lrbp}{\ensuremath{\,\MNSlessdot'\,}}
\newcommand{\lreqb}{\ensuremath{\,\MNSleqdot\,}}
\newcommand{\lreqbp}{\ensuremath{\,\MNSleqdot'\,}}
\newcommand{\lreqd}{\ensuremath{\MNSleqdot}}

\SetKwFunction{labelsf}{labels}
\SetKwFunction{fl}{mj}
\SetKwFunction{neph}{neph}
\SetKwFunction{unc}{unc}
\SetKwFunction{lpf}{unc}
\SetKwFunction{lsf}{neph}
\SetKwFunction{descf}{desc}
\SetKwFunction{first}{first}
\SetKwFunction{lmd}{lmd}
\SetKwFunction{min}{min}
\SetKwFunction{max}{max}
\SetKwFunction{ml}{ml}

\newcommand{\fullv}[2]{#2}
\title{Profile Trees for B\"uchi Word Automata, with Application to Determinization%
\thanks{Work supported in part by NSF grants CNS 1049862 and CCF-1139011, by NSF
Expeditions in Computing project ``ExCAPE: Expeditions in Computer Augmented
Program Engineering'', by a gift from Intel, by BSF grant 9800096, and by a
stipend from Trinity University.
A full version, with appendices and missing proofs, is available at
\url{http://www.cs.trinity.edu/\~sfogarty/papers/gandalf13rj.pdf}
}}

\def\titlerunning{Profile Trees for B\"uchi Word Automata, with Application to Determinization}

\author{Seth Fogarty
\institute{Computer Science Department\\Trinity University}
\and
Orna Kupferman
\institute{School of Computer Science and Engineering\\Hebrew University of Jerusalem}
\and
Moshe Y. Vardi\ \ \ \ 
\institute{Department of Computer Science\\Rice University}\ \ \ \ 
\and
Thomas Wilke\ \ \ \ \ \ 
\institute{Institut f\"ur Informatik\\Christian-Albrechts-Universit\"at zu Kiel}\ \ \ \ \ \ 
}

\def\authorrunning{S. Fogarty, O. Kupferman, M.Y. Vardi, \& Th. Wilke}

\begin{document}
\maketitle

\begin{abstract}

The determinization of \buchi automata is a celebrated problem, with applications in synthesis,
probabilistic verification, and multi-agent systems.  Since the 1960s, there has been a steady
progress of constructions: by McNaughton, Safra, Piterman, Schewe, and others.  Despite the
proliferation of solutions, they are all essentially ad-hoc constructions, with little theory
behind them other than proofs of correctness. Since Safra, all optimal constructions employ
trees as states of the deterministic automaton, and transitions between states are defined
operationally over these trees.  The operational nature of these constructions complicates
understanding, implementing, and reasoning about them, and should be contrasted with
complementation, where a solid theory in terms of automata run \DAGs underlies modern constructions.

In 2010, we described a \emph{profile}-based approach to B\"uchi complementation, where a profile is
simply the history of visits to accepting states. We developed a structural theory of profiles and
used it to describe a complementation construction that is deterministic in the limit.  Here we
extend the theory of profiles to prove that every run \DAG contains a \emph{profile tree} with at
most a finite number of infinite branches. We then show that this property provides a theoretical
grounding for a new determinization construction where macrostates are doubly preordered sets of
states.  In contrast to extant determinization constructions, transitions in the new construction
are described declaratively rather than operationally.
\end{abstract}

\section{Introduction}

\buchi automata were introduced in the context of decision problems for second-order arithmetic
\cite{Buc62}.  These automata constitute a natural generalization of automata over finite words to
languages of infinite words.  Whereas a run of an automaton on finite words is accepting if the run
ends in an accepting state, a run of a \buchi automaton is accepting if it visits an accepting state
infinitely often.

Determinization of nondeterministic automata is a fundamental problem in automata theory, going back
to \cite{RS59}.  Determinization of \buchi automata is employed in many applications, including
synthesis of reactive systems \cite{PR89a}, verification of probabilistic systems
\cite{CY95,Var85b}, and reasoning about multi-agent systems \cite{AHK02}.  Nondeterministic automata
over finite words can be determinized with a simple, although exponential, \emph{subset
construction} \cite{RS59}, where a state in the determinized automaton is a set of states of the
input automaton.  Nondeterministic \buchi automata, on the other hand, are not closed under
determinization, as deterministic \buchi automata are strictly less expressive than their
nondeterministic counterparts \cite{Lan69}.  Thus, a determinization construction for \buchi
automata must result in automata with a more powerful acceptance condition, such as Muller
\cite{McN66}, Rabin \cite{Saf88}, or parity conditions \cite{KW08,Pit06}.

The first determinization construction for \buchi automata was presented by McNaughton, with a
doubly-exponential blowup \cite{McN66}. In 1988, Safra introduced a singly exponential construction
\cite{Saf88}, matching the lower bound of $n^{O(n)}$ \cite{Lod99}. Safra's construction encodes a
state of the determinized automaton as a labeled tree, now called a \emph{Safra tree}, of sets of
states of the input \buchi automaton.  Subsequently, Safra's construction was improved by Piterman,
who simplified the use of tree-node labels \cite{Pit06}, and by Schewe, who moved the acceptance
conditions from states to edges \cite{Sch09b}.  In a separate line of work, Muller and Schupp
proposed in 1995 a different singly exponential determinization construction, based on
\emph{Muller-Schupp trees} \cite{MS95}, which was subsequently simplified by \kahler and Wilke
\cite{KW08}. 

Despite the proliferation of \buchi determinization constructions, even in their improved and
simplified forms all constructions are essentially ad-hoc, with little theory behind them other than
correctness proofs. These constructions rely on the encoding of determinized-automaton states as
finite trees.  They are operational in nature, with transitions between determinized-automaton
states defined ``horticulturally,'' as a sequence of operations that grow trees and then prune them
in various ways.  The operational nature of these constructions complicates understanding,
implementing, and reasoning about them \cite{ATW05,THB95}, and should be contrasted with
complementation, where an elegant theory in terms of automata run \DAGs underlies modern
constructions \cite{FKV06,KV01c,Sch09}.  In fact, the difficulty of determinization has motivated
attempts to find determinization-free decision procedures \cite{KV05c} and works on determinization
of fragments of LTL \cite{KE12}.

In a recent work \cite{FKVW11}, we introduced the notion of \emph{profiles} for nodes in the run \DAG.
We began by labeling accepting nodes of the \DAG by $1$ and non-accepting nodes by $0$, essentially
recording visits to accepting states.  The profile of a node is the lexicographically \emph{maximal}
sequence of labels along paths of the run \DAG that lead to that node.  Once profiles and a
lexicographic order over profiles were defined, we removed from the run \DAG edges that do not
contribute to profiles.  In the pruned run \DAG, we focused on lexicographically maximal runs.  This
enabled us to define a novel, profile-based \buchi complementation construction that yields
\emph{deterministic-in-the-limit} automata: one in which every accepting run of the complementing
automaton is eventually deterministic \cite{FKVW11} A state in the complementary automaton is a set
of states of the input nondeterministic automaton, augmented with the preorder induced by profiles.
Thus, this construction can be viewed as an augmented subset construction.

In this paper, we develop the theory of profiles further, and consider the equivalence classes of
nodes induced by profiles, in which two nodes are in the same class if they have the same profile.
We show that profiles turn the run \DAG into a \emph{profile tree}: a binary tree of bounded
width over the equivalence
classes. The profile tree affords us a novel singly exponential \buchi determinization construction.
In this profile-based determinization construction, a state of the determinized automaton is a set
of states of the input automaton, augmented with \emph{two} preorders induced by profiles.  Note
that while a Safra tree is finite and encodes a single level of the run \DAG, our profile tree is
infinite and encodes the entire run \DAG, capturing the accepting or rejecting nature of all paths.
Thus, while a state in a traditional determinization construction corresponds to a Safra tree, a
state in our deterministic automaton corresponds to a single level in the profile tree.

Unlike previous \buchi determinization constructions, transitions between states of the determinized
automaton are defined declaratively rather than operationally. We believe that the declarative
character of the new construction will open new lines of research on \buchi determinization.  For
\buchi complementation, the theory of run \DAGs \cite{KV01c} led not only to tighter constructions
\cite{FKV06,Sch09}, but also to a rich body of work on heuristics and optimizations
\cite{DR07,FV10}.  We foresee analogous developments in research on \buchi
determinization.

\section{Preliminaries}

This section introduces the notations and definitions employed in our analysis\hide{ of \buchi
automata}.

\subsection{Relations on Sets}
Given a set $R$, a binary relation $\leq$ over $R$ is a \emph{preorder} if $\leq$ is reflexive and
transitive.  A \emph{linear preorder} relates every two elements: for every $r_1, r_2 \in R$ either
$r_1 \leq r_2$, or $r_2 \leq r_1$, or both.  A relation is \emph{antisymmetric} if $r_1 \leq r_2$
and $r_2 \leq r_1$ implies $r_1 = r_2$. A preorder that is antisymmetric is a \emph{partial order}.
A linear partial order is a \emph{total order}. Consider a partial order $\leq$. If for every $r \in
R$, the set $\set{r' \mid r' \leq r}$ of smaller elements is totally ordered by $\leq$, then we say
that $\leq$ is a \emph{tree order}.  The equivalence class of $r \in R$ under $\leq$,  written
$[r]$, is $\set{r' \mid r' \leq r \text{ and } r \leq r'}$. The equivalence classes under a linear
preorder form a totally ordered partition of $R$.  Given a set $R$ and linear preorder $\leq$ over
$R$, define the minimal elements of $R$ as $\min_{\leq}(R)=\set{r_1 \in R\mid ~r_1 \leq r_2 \text{
for all }r_2 \in R}$.  Note that $\min_\leq(R)$ is either empty or an equivalence class under
$\leq$.  Given a non-empty set $R$ and a total order $\leq$, we instead define $\min_\leq(R)$ as the
the unique minimal element of $R$.

Given two finite sets $R$ and $R'$ where $\abs{R} \leq \abs{R'}$, a linear preorder $\leq$ over $R$,
and a total order $<'$ over $R'$, define the \emph{$\zug{\leq,<'}$-minjection} from $R$ to $R'$ to
be the function $\fl$ that maps all the elements in the $k$-th equivalence class of $R$ to the
$k$-th element of $R'$. The number of equivalence classes is at most $\abs{R}$, and thus at most
$\abs{R'}$. If $\leq$ is also a total order, than the $\zug{\leq,<'}$-minjection is also an
injection.  

\begin{example}
Let $R=\Q$ and $R'=\Z$ be the sets of rational numbers and integers, respectively.  Define the
linear preorder $\leq_1$ over $\Q$ by $x \leq_1 x'$ iff $\lfloor x \rfloor \leq \lfloor x' \rfloor$,
and the total order $<_2$ over $\Z$ by $x <_2 x'$ if $x < x'$. Then, the
$\zug{\leq_1,<_2}$-minjection from $\Q$ to $\Z$ maps a rational number $x$ to $\lfloor x \rfloor$.
\end{example}

\subsection{$\omega$-Automata}
A \emph{nondeterministic $\omega$-automaton} is a tuple $\A=\zug{\Sigma, Q, Q^{in}, \rho, \alpha}$,
where $\Sigma$ is a finite alphabet, $Q$ is a finite set of states, $Q^{in} \subseteq Q$ is a set of
initial states, $\rho \colon Q \times \Sigma \to 2^Q$ is a nondeterministic transition relation, and
$\alpha$ is an acceptance condition defined below.  An automaton is \emph{deterministic} if
$\abs{Q^{in}} = 1$ and, for every $q \in Q$ and $\sigma \in \Sigma$, we have
$\abs{\rho(q,\sigma)}=1$. For a function $\delta \colon Q \times \Sigma \to 2^Q$, we
lift $\delta$ to sets $R$ of states in the usual fashion: $\delta(R,\sigma) = \bigcup_{r \in R}
\delta(r,\sigma)$.  Further, we define the inverse of $\delta$, written $\delta^{-1}$, to be
$\delta^{-1}(r,\sigma)=\set{q \mid r \in \delta(q,\sigma)}$. 

A \emph{run} of an $\omega$-automaton $\A$ on a word $w=\sigma_0\sigma_1\cdots \in
\Sigma^\omega$ is an infinite sequence of states $q_0,q_1,\ldots\in Q^\omega$ such that $q_0 \in Q^{in}$
and, for every $i \geq 0$, we have that $q_{i+1} \in \rho(q_i, \sigma_i)$. Correspondingly, a
\emph{finite run} of $\A$ to $q$ on $w=\sigma_0\cdots \sigma_{n-1} \in \Sigma^*$ is a finite
sequence of states $p_0,\ldots,p_n$ such that $p_0\in Q^{in}$, $p_n=q$, and for every $0 \leq i <
n$ we have $p_{i+1} \in \rho(p_i, \sigma_i)$. 

The acceptance condition $\alpha$ determines if a run is \emph{accepting}.  If a run is not
accepting, we say it is \emph{rejecting}.  A word $w \in \Sigma^\omega$ is accepted by $\A$ if there
exists an accepting run of $\A$ on $w$.  The words accepted by $\A$ form the \emph{language} of
$\A$, denoted by $L(\A)$.  For a \emph{\buchi automaton}, the acceptance condition is a set of
states $F \subseteq Q$, and a run $q_0,q_1,\ldots$ is accepting iff $q_i \in F$ for infinitely many $i$'s.  For
convenience, we assume $Q^{in} \cap F = \emptyset$.  For a \emph{Rabin automaton}, the acceptance
condition is a sequence $\zug{G_0,B_0},\ldots,\zug{G_k,B_k}$ of pairs of sets of states.
Intuitively, the sets $G$ are ``good'' conditions, and the sets $B$ are ``bad'' conditions. A run
$q_o,q_1,\ldots$ is
accepting iff there exists $0 \leq j\leq k$ so that $q_i \in G_j$ for infinitely many $i$'s, while
$q_i \in B_j$ for only finitely many $i$'s.  Our focus in this paper is on nondeterministic \buchi
automata on words (NBW) and deterministic Rabin automata on words (DRW).

\subsection{Safra's Determinization Construction}
This section presents Safra's determinization construction, using the exposition in
\cite{Pit06}.  Safra's construction takes an NBW and constructs an equivalent DRW.  Intuitively, a
state in this construction is a tree of subsets.  Every node in the tree is labeled by the states it
follows. The label of a node is a strict superset of the union of labels of its descendants, and the
labels of siblings are disjoint. Children of a node are ordered by ``age''.  Let $\A=\zug{\Sigma, Q,
Q^{in}, \rho, F}$ be an NBW, $n=\abs{Q}$, and $V=\set{0,\ldots, n-1}$.

\begin{definition}\rm{\cite{Pit06}}
A \emph{Safra tree} over $\A$ is a tuple $t=\zug{N,r,p,\psi,l,G,B}$ where:
\begin{compactitem}
\item $N \subseteq V$ is a set of nodes.
\item $r \in N$ is the root node.
\item $p : (N \setminus \set{r}) \rightarrow N$ is the parent function over $N \setminus \set{r}$.
\item $\psi$ is a partial order defining 'older than' over siblings.
\item $l : N \rightarrow 2^{Q}$ is a labeling function from nodes to
non-empty sets of states. The label of every node is a proper superset of the
union of the labels of its sons. The labels of two siblings are disjoint.
\item $G, B \subseteq V$ are two disjoint subsets of $V$. 
\end{compactitem}
\end{definition}

The only way to move from one Safra tree to the next is through a sequence of ``horticultural''
operations, growing the tree and then pruning it to ensure that the above invariants hold.

\begin{definition}\label{Def:Safra_Condition}
Define the DRW $D^S(\A)=\zug{\Sigma,Q_S,\rho_S,t_0,\alpha }$ where:
\begin{compactitem}
\item $Q_S$ is the set of Safra trees over $\A$. 
\item $t_0=\zug{\set{0},0,\emptyset,\emptyset,l_0,\emptyset,\set{1,\ldots,n-1}}$ where $l_0(0)=Q^{in}$
\item For $t=\zug{N,r,p,\psi,l,G,B} \in Q_S$ and $\sigma \in \Sigma$, the tree
$t'=\rho_S(t,\sigma)$ is the result of the following sequence of operations.
We temporarily use a set $V'$ of names disjoint from $V$. Initially, let
$t'=\zug{N',r',p',\psi',l',G',B'}$ where $N'=N$, $r'=r$, $p'=p$,
$\psi'=\psi$, $l'$ is undefined, and $G'=B'=\emptyset$.
\begin{compactenum}
\item For every $v \in N'$, let $l'(v)=\rho(l(v),\sigma)$.
\item For every $v \in N'$ such that $l'(v) \cap F \neq \emptyset$, create a new node $v' \in V'$
where: $p(v')=v$; $l'(v')=l'(v) \cap F$; and for every $w' \in V'$ where $p(w')=v$ add $(w',v')$ to
$\psi$.
\item For every $v \in N'$ and $q \in l'(v)$, if there is a $w \in N'$ such that $(w,v) \in \psi$ and
$q \in l'(w)$, then remove $q$ from $l'(v)$ and, for every descendant $v'$ of $v$, remove $q$ from
$l'(v')$. 
\item Remove all nodes with empty labels.
\item For every $v \in N'$, if $l'(v)=\bigcup \set{l'(v') \mid p'(v')=v}$ remove all children of
$v$, add $v$ to $G$.
\item Add all nodes in $V \setminus N'$ to $B$.
\item Change the nodes in $V'$ to unused nodes in $V$.
\end{compactenum}
\item $\alpha = \set{\zug{G_0,B_0},\ldots,\zug{G_{n-1},B_{n-1}}}$, where:
\begin{compactitem}
\item $G_i = \set{\zug{N,r,p,\psi,l,G,B} \in Q_S \mid i \in G}$
\item $B_i = \set{\zug{N,r,p,\psi,l,G,B} \in Q_S \mid i \in B}$
\end{compactitem}
\end{compactitem}
\end{definition}

\begin{theorem}{\rm \cite{Saf88}}
For an NBW $\A$ with $n$ states, $L(D^S(\A))\!=\!L(\A)$ and $D^S(\A)$ has $n^{O(n)}$ states.
\end{theorem}

\section{From Run \DAGs to Profile Trees}\label{Sect:Profiles}

In this section, we present a framework for simultaneously reasoning about all runs of a \buchi
automaton on a word.  We use a \DAG to encode all possible runs, and give each node in this \DAG a
profile based on its history. The lexicographic order over profiles induces a preorder $\preceq_i$
over the nodes on level $i$ of the run \DAG. Using $\preceq_i$, we prune the edges of the run \DAG,
and derive a binary tree of bounded width.  Throughout this paper we fix an NBW $\A
=\zug{\Sigma,Q,Q^{in},\rho,F}$ and an infinite word $w = \sigma_0\sigma_1\cdots$.

\subsection{Run \DAGs and Profiles}
The runs of $\A$ on $w$ can be arranged in an infinite \DAG $\G=\zug{V,E}$, where 
\begin{itemize}
\item $V \subseteq Q \times \N$ is such that $\rzug{q,i} \in V$ iff there is a finite run of $\A$ 
to $q$ on $\sigma_0\cdots \sigma_{i-1}$.
\item $E \subseteq \bigcup_{i \geq 0} (Q \stimes \set{i}) \!\times\! (Q
\stimes \set{i\splus 1})$ is such that $E(\rzug{q,i},\rzug{q',i \splus 1})$ iff $\rzug{q,i} \in V$ and
$q' \in \rho(q, \sigma_i)$.
\end{itemize}
The \DAG $\G$, called the \emph{run \textsc{dag} of $\A$ on $w$}, embodies all possible runs
of $\A$ on $w$. We are primarily concerned with \emph{initial paths} in $\G$: paths that start in
$Q^{in} \times \set{0}$.  A node $\rzug{q,i}$ is an $F$-node if $q \in F$, and a path in
$\G$ is \emph{accepting} if it is both initial and contains infinitely many $F$-nodes.  An
accepting path in $\G$ corresponds to an accepting run of $\A$ on $w$. If $\G$ contains an
accepting path, we say that $\G$ is \emph{accepting}; otherwise it is \emph{rejecting}.  
Let $G'$ be a sub-$\DAG$ of $G$. For $i \geq 0$, we refer to the nodes in $Q \times \{i\}$ as
{\em level $i$\/} of $G'$. Note that a node on level $i+1$ has edges only from nodes on level $i$.
We say that $G'$ has \emph{bounded width of degree $c$} if every level in $G'$ has at most $c$
nodes. By construction, $\G$ has bounded width of degree $\abs{Q}$.

Consider the run \DAG $\G=\zug{V,E}$ of $\A$ on $w$. Let $f \colon V \to \set{0,1}$ be such that
$f(\rzug{q,i}) = 1$ if $q \in F$ and $f(\rzug{q,i}) = 0$ otherwise. Thus, $f$ labels $F$-nodes by
$1$ and all other nodes by $0$. The \emph{profile} of a path in $\G$ is the sequence of labels of
nodes in the path.  We define the profile of a node to be the lexicographically maximal profile of
all initial paths to that node.  Formally, the profile of a finite
path $b=v_0,v_1,\ldots,v_n$ in $\G$, written $h_b$, is $f(v_0)f(v_1)\cdots f(v_n)$, and the profile
of an infinite path $b=v_0,v_1,\ldots$ \linebreak[3] is $h_b=f(v_0)f(v_1)\cdots$.  Finally, the profile of a node
$v$, written $h_v$, is the lexicographically maximal element of $\set{h_b\mid b \text{ is an initial
path to }v}$.

The lexicographic order of profiles induces a linear preorder over nodes on
every level of $\G$.  We define a sequence of linear preorders $\preceq_i$ over
the nodes on level $i$ of  $\G$ as follows.  For nodes $u$ and $v$ on level
$i$, let $u \prec_i v$ if $h_u < h_v$, and $u \approx_i v$ if $h_u = h_v$.  We
group nodes by their equivalence classes under $\preceq_i$. Since the final
element of a node's profile is $1$ if and only if the node is an $F$-node, all nodes in an
equivalence class agree on membership in $F$. Call an equivalence class an
$F$-class when all members are $F$-nodes, and a non-$F$-class when none of its
members are $F$-nodes. When a state can be reached by two finite runs, a node
will have multiple incoming edges in $\G$.  We now remove from $\G$ all edges
that do not contribute to profiles.  Formally, define the pruned run \DAG
$\Gprime=\zug{V,E'}$ where $E' = \set{\rzug{u,v} \in E \mid \text{for every $u'
\in V$, if $\zug{u',v} \in E$ then $u' \preceq_{\abs{u}} u$}}$. Note that the
set of nodes in $\G$ and $\Gprime$ are the same, and that an edge is removed
from $E'$ only when there is another edge to its destination.

Lemma~\ref{Gprime_Captures_Profiles} states that, as we have removed only edges that do not
contribute to profiles, nodes derive their profiles from their parents in $\Gprime$.

\begin{lemma}{\rm \cite{FKVW11}} \label{Gprime_Captures_Profiles}
For two nodes $u$ and $u'$ in $V$, if $\zug{u,u'} \in E'$, then $h_{u'} = h_u0$ or $h_{u'} = h_u1$.
\end{lemma}

While nodes with different profiles can share a child in $\G$, Lemma~\ref{Gprime_Parents_Equivalent}
precludes this in $\Gprime$. 

\begin{restatable}{lemma}{lemGprimeParentsEquivalent}\label{Gprime_Parents_Equivalent}
Consider nodes $u$ and $v$ on level $i$ of $\Gprime$ and nodes $u'$ and $v'$ on level $i+1$ of
$\Gprime$. If $\zug{u,u'} \in E'$, $\zug{v,v'} \in E'$, and
$u' \approx_{i+1} v'$, then $u \approx_i v$.
\end{restatable}
\begin{proof}
Since $u' \approx_{i+1} v'$, we have $h_{u'}=h_{v'}$. If ${u'}$ is an $F$-node, then $v'$ is 
an $F$-node and the last letter in both $h_{u'}$ and $h_{v'}$ is $1$. By Lemma
\ref{Gprime_Captures_Profiles} we have $h_{u}1 = h_{u'} = h_{v'} = h_{v}1$. If ${u'}$ and $v'$ are
non-$F$-nodes, then we have $h_{u}0 = h_{u'} = h_{v'} =
h_{v}0$. In either case, $h_{u}=h_{v}$ and ${u} \approx_{i} v$.
\end{proof}

Finally, we have that $\Gprime$ captures the accepting or rejecting nature of $\G$. This result 
was employed to provide deterministic-in-the-limit complementation in \cite{FKVW11}

\begin{theorem}{\rm \cite{FKVW11}} \label{Lexicographic_Edge_Pruning}
The pruned run \DAG $\Gprime$ of an NBW $\A$ on a word $w$ is accepting iff $\A$ accepts
$w$. 
\end{theorem}

\subsection{The Profile Tree}\label{Sect:TGP}
Using profiles, we define the \emph{profile tree} $\TGP$\!, which we show to be a binary tree of
bounded width that captures the accepting or rejecting nature of the pruned run \DAG $\Gprime$. The
nodes of $\TGP$ are the equivalence classes $\set{[u]\mid u \in V}$ of $\Gprime=\zug{V,E'}$.  To remove
confusion, we refer to the nodes of $\TGP$ as \emph{classes} and use and $U$ and $W$ for classes in
$\TGP$, while reserving $u$ and $v$ for nodes in $\G$ or $\Gprime$.  The edges in $\TGP$ are induced
by those in $\Gprime$ as expected: for an edge $\zug{u,v} \in E'$, the class $[v]$ is the child of $[u]$
in $\TGP$.  A class $W$ is a \emph{descendant} of a class $U$ if there is a, possibly empty, path
from $U$ to $W$.

\begin{restatable}{theorem}{thmTGPBinTree}\label{TGPBinTree}
The profile tree $\TGP$ of an $n$-state NBW $\A$ on an infinite word $w$ is a binary tree whose
width is bounded by $n$.
\end{restatable}
\begin{proof}
That \TGP has bounded width follows from the fact that a class on level $i$ contains at least one
node on level $i$ of \G, and \G is of bounded width of degree $\abs{Q}$. To prove that every class has one
parent, for a class $W$ let $U = \set{u\mid \mbox{there is } v \in W \mbox{ such that } \zug{u,v}
\in E'}$.  Lemma \ref{Gprime_Parents_Equivalent} implies that $U$ is an equivalence class, and is
the sole parent of $W$.  To show that $\TGP$ has a root, note that as $Q^{in} \cap F = \emptyset$,
all nodes on the first level of $\G$ have profile $0$, and every class descends from this class of
nodes with profile $0$.  Finally, as noted Lemma~\ref{Gprime_Captures_Profiles} entails that a class
$U$ can have at most two children: the class with profile $h_U1$, and the class with profile $h_U0$.
Thus \TGP is binary.
\end{proof}

A \emph{branch} of $\TGP$ is a finite or infinite initial path in $\TGP$.  Since $\TGP$ is a tree,
two branches share a prefix until they \emph{split}.  An infinite branch is \emph{accepting} if it
contains infinitely many $F$-classes, and \emph{rejecting} otherwise. An infinite rejecting
branch must reach a suffix consisting only of non-$F$-classes.  \hide{Note that if $U'$ is a descendant of
both $U$ and $W$, either $U$ is a descendant of $W$, or $W$ is a descendant of $U$.} A class $U$ is
called \emph{finite} if it has finitely many descendants, and a finite class $U$ {\em dies out} on
level $k$ if it has a descendant on level $k-1$, but none on level $k$. Say $\TGP$ is
\emph{accepting} if
it contains an accepting branch, and \emph{rejecting} if all branches are rejecting.

As all members of a class share a profile, we define the profile $h_U$ of a class $U$ to be $h_u$ for
some node $u \in U$.  We extend the function $f$ to classes, so that $f(U)=1$ if~$U$ is an $F$-class, and
$f(U)=0$ otherwise.  We can then define the profile of an infinite branch $b=U_0,U_1,\ldots$ to be
$h_b=f(U_0)f(U_1)\cdots$. For two classes $U$ and $W$ on level $i$, we say that $U \prec_i W$ if $h_U < h_W$.  For
two infinite branches $b$ and $b'$, we say that $b \prec b'$ if $h_b < h_{b'}$.  Note that $\prec_i$
is a total order over the classes on level $i$, and that $\prec$ is a total order over
the set of infinite branches.

As proven above, a class $U$ has at most two children: the class of $F$-nodes with profile $h_U1$,
and the class of non-$F$-nodes with profile $h_U0$.  We call the first class the $F$-child of $U$,
and the second class the non-$F$-child of $U$.  While the \DAG $\Gprime$ can have infinitely many infinite
branches, bounding the width of a tree also bounds the number of infinite branches it may
have. 

\begin{corollary}\label{TGP_Finite}
The profile tree $\TGP$ of an NBW $\A$ on an infinite word $w$ has a finite number of infinite
branches.
\end{corollary}

\showfigure{
\begin{figure}[tb]
\begin{centering}
\subfloat[An automaton]{
\raisebox{0in}{
\begin{tikzpicture} [->,auto,node distance=0.6cm,line width=0.4mm]
\node[state,initial,inner sep=1pt] (q) {$q$};
\node[state,accepting] (p) [right=of q] {$p$};
\path 	(q) edge [loop above] node {a} (q)
			edge [bend left] node {a} (p)
		(p) edge [bend left] node {b} (q)
			edge [loop above] node {a,b} (p);
\end{tikzpicture}
}}
\!\!\!\!\!\!\!\!\!\!\!\!\!\!\!\!\!\!\!\!\!\!\!\!\subfloat[\TGP for automaton (a) on $ab^\omega$.]
{
\begin{tikzpicture} [node distance=0.4cm,line width=0.4mm]
\node[state] (p0) [label=right:
	$\begin{array}{l} h\eq 0\\ \protect\labelsf \eq \{\}\\\protect\gl \eq 0 \end{array}$
	] {${\zug{q,0}}$};

\node[state] (q1) [below=of p0,label=left:
	$\begin{array}{r} h\eq 00\\ \protect\labelsf \eq \{0\}\\\protect\gl \eq 0 \end{array}$
	] {${\zug{q,1}}$}
		edge [] (p0);
\node[state,accepting] (p1) [right=of q1, label=right:
	$\begin{array}{l} h\eq 01\\ \protect\labelsf \eq \{\}\\\protect\gl \eq 1 \end{array}$
	] {${\zug{p,1}}$}
		edge [] (p0);

\node[state] (q2) [below=of p1,label=left:
	$\begin{array}{r} h\eq 010\\ \protect\labelsf \eq \set{0,1}\\\protect\gl \eq 0 \end{array}$
	] {${\zug{q,2}}$}
		edge [] (p1);
\node[state,accepting] (p2) [right=of q2, label=right:
	$\begin{array}{l} h\eq 011\\ \protect\labelsf \eq \{\}\\\protect\gl \eq 2 \end{array}$
	] {${\zug{p,2}}$}
		edge [] (p1);

\node[state] (q3) [below=of p2,label=left:
	$\begin{array}{r} h\eq 0110\\ \protect\labelsf \eq \set{0,1,2}\\\protect\gl \eq 0 \end{array}$
	] {${\zug{q,3}}$}
		edge [] (p2);
\node[state,accepting] (p3) [right=of q3, label=below:
	$~~~~~~~~~\begin{array}{l} h\eq 0111\\ \protect\labelsf \eq \{\}\\\protect\gl \eq 3 \end{array}$
	] {${\zug{p,3}}$}
		edge [] (p2);

\draw[fill] (3,-6.00) circle (0.5mm);
\draw[fill] (3,-6.25) circle (0.5mm);
\draw[fill] (3,-6.50) circle (0.5mm);
\end{tikzpicture}
}
\caption{An automaton and tree of classes. Each class is a singleton set, brackets are omitted for
brevity.  $F$-classes are circled twice.  Each class is labeled with its profile $h$, as well as the
set $\protect\labelsf$ and the global label $\gl$ as defined in Section~\ref{Sect:Global_Labeling}. 
\label{Fig:TGP}}
\end{centering}
\end{figure}
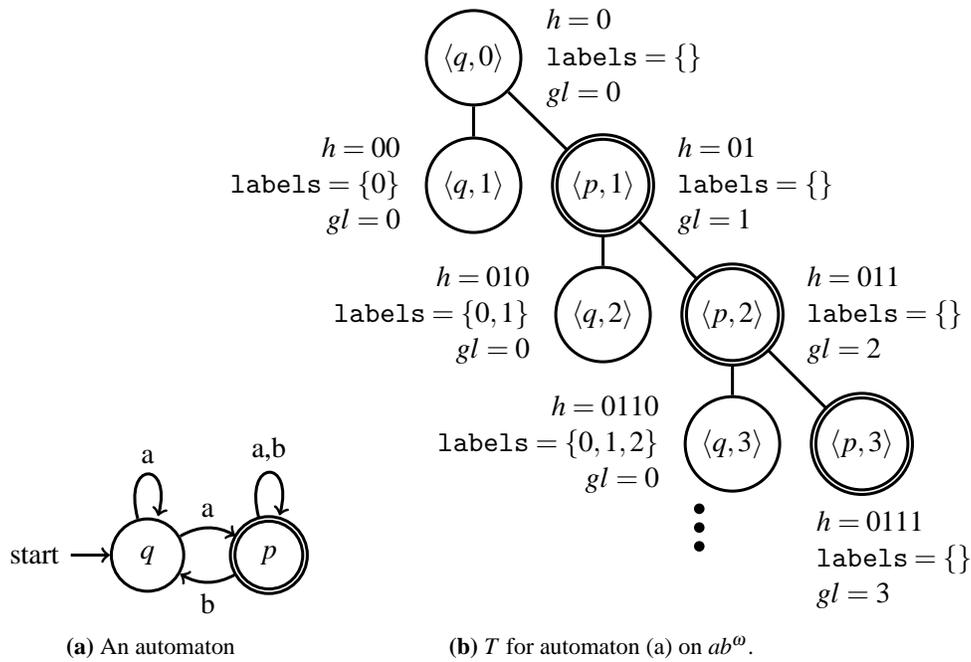
}

\begin{example}
Consider, for example, the NBW in Figure~\ref{Fig:TGP}.(a) and the
first four levels of a tree of equivalence classes in Figure~\ref{Fig:TGP}.(b).  This tree
corresponds to all runs of the NBW on the word $ab^\omega$.  There is only one infinite branch,
$\set{\zug{q,0}},\set{\zug{p,1}},\set{\zug{p,2}},\ldots$, which is accepting.  The set of labels and
the global labeling $\gl$ are explained below, in Section~\ref{Sect:Global_Labeling}.
\end{example}

We conclude this section with Theorem~\ref{TGP_Captures}, which enables us to reduce the search for
an accepting path in $\Gprime$ to a search for an accepting branch in $\TGP$. 

\pagebreak[4]
\begin{restatable}{theorem}{thmTGPCaptures}\label{TGP_Captures}
The profile tree $\TGP$ of an NBW $\A$ on an infinite word $w$ is accepting iff $\A$ accepts $w$.
\end{restatable}
\begin{proof}
If $w \in L(\A)$, then by Theorem~\ref{Lexicographic_Edge_Pruning} we have that $\Gprime$ contains
an accepting path $u_0,u_1,\ldots$. This path gives rise to an accepting branch
$[u_0],[u_1],\ldots$ in $\TGP$. In the other direction, if $\TGP$ has an accepting branch
$U_0,U_1,\ldots$, consider the infinite subgraph of $\Gprime$ consisting only of the nodes in $U_i$, for
$i>0$.  For every $i>0$ there exists $u_i \in U_i$ and $u_{i+1} \in U_{i+1}$ so that
$E'(u_i,u_{i+1})$.  Because no node is orphaned in $\Gprime$, Lemma~\ref{Gprime_Parents_Equivalent}
implies that every node in $U_{i+1}$ has a parent in $U_{i}$, thus this subgraph is connected. As
each node has degree of as most $n$, \konig's Lemma implies that there is an infinite initial path
$u_0,u_1,\ldots$ through this subgraph.  Further, at every level $i$ where $U_i$ is an $F$-class, we
have that $u_i \in F$, and thus this path is accepting and $w \in L(\A)$.
\end{proof}

\section{Labeling}\label{Sect:Labeling}

In this section we present a method of deterministically labeling the classes in $\TGP$ with
integers, so we can determine if $\TGP$ is accepting by examining the labels.  Each label $m$
represents the proposition that the lexicographically minimal infinite branch through the
first class labeled with $m$ is accepting. On each level we give the label $m$ to the
lexicographically minimal descendant, on any branch, of this first class labeled with $m$.  We
initially allow the use of global information about \TGP and an unbounded number of labels. We then
show how to determine the labeling using bounded information about each level of $\TGP$, and how to
use a fixed set of labels.

\subsection{Labeling $\TGP$}\label{Sect:Global_Labeling}
We first present a labeling that uses an unbounded number of labels and global
information about $\TGP$.  We call this labeling the \emph{global labeling}, and denote it with
$\gl$.  For a class $U$ on level $i$ of $\TGP$, and a class $W$ on level $j$, we say that $W$ is
{\em before\/} $U$ if $j < i$ or $j=i$ and $W \prec_i U$. For each label $m$, we refer to the first
class labeled $m$ as $\first(m)$. Formally, $U=\first(m)$ if $U$ is labeled $m$ and, for all classes
$W$ before $U$, the label of $W$ is not $m$.  We define the labeling function $\gl$ inductively over
the nodes of $\TGP$.  For the initial class $U_0=\set{\rzug{q,0}\mid q \in Q^{in}}$ with profile
$0$, let $\gl(U_0)=0$. 

Each label $m$ follows the lexicographically minimal child of $\first(m)$ on every level.  When a
class with label $m$ has two children, we are not certain which, if either, is part of an infinite
branch. We are thus conservative, and follow the non-$F$-child. If the non-$F$-child dies out, we
revise our guess and move to a descendant of the $F$-child.  For a label $m$ and level $i$, let the
\emph{lexicographically minimal descendant} of $m$ on level $i$, written $\lmd(m,i)$, be
$\min_\preceq(\{W\mid W\text{ is a descendant of $\first(m)$ on level $i$}\})$:
the class with the minimal profile among all the descendants of $\first(m)$ on level
$i$.  For a class $U$ on level $i$, define $\labelsf(U)=\set{m \mid U=\lmd(m,i)}$ as the set of
valid labels for $U$. When labelling $U$, if $U$ has more than one valid label, we give it
the smallest label, which corresponds to the earliest ancestor. If $\labelsf{U}$ is empty, $U$ is
given an unused label one greater than the maximum label occurring earlier in $\TGP$.

\begin{definition}\label{Def:Labels}
$\gl(U) = 
\begin{cases}
\min(\labelsf(U)) & \text {if } \labelsf(U) \neq \emptyset,\\
\max(\{\gl(W)\mid W \text{ is before } U\}) +1  & \text{if }\labelsf(U) = \emptyset.\\
\end{cases}
$
\end{definition}

Lemma \ref{Label_Props} demonstrates that every class on a level gets a unique label, and that
despite moving between nephews the labeling adheres to branches in the tree. \fullv{The proof is reserved
for Appendix~\ref{App:Proofs}.}

\begin{restatable}{lemma}{lemLabelProps}\label{Label_Props}
For classes $U$ and $W$ on level $i$ of $\TGP$, it holds that:
\begin{compactenum}
\item\label{L_Distinct} If $U \neq W$ then $\gl(U) \neq \gl(W)$.
\item\label{L_Descends} $U$ is a descendant of $\first(\gl(U))$.
\item\label{L_Preceq} If $U$ is a descendant of $\first(\gl(W))$, then $W \preceq_i U$.
Consequently, if $U \prec_i W$, then $U$ is not a descendant of $\first(\gl(W))$.
\item\label{First_F} $\first(\gl(U))$ is the root or an $F$-class with a sibling.
\item\label{Persistent_Labels} If $U \neq \first(\gl(U))$, then there is a class on level $i-1$
that has label $\gl(U)$.
\item\label{Less_Earlier} If $\gl(U) < \gl(W)$ then $\first(\gl(U))$ is before $\first(\gl(W))$. 
\end{compactenum}
\end{restatable}

As stated above, the label $m$ represents the proposition that the lexicographically minimal
\emph{infinite} branch going through $\first(m)$ is accepting. Every time we pass through an
$F$-child, this is evidence towards this proposition.  Recall that when a class with label $m$ has
two children, we initially follow the non-$F$-child. If the non-$F$-child dies out, we revise our
guess and move to a descendant of the $F$-child.  Thus revising our guess indicates that at an
earlier point the branch did visit an $F$-child, and also provides evidence towards this
proposition.  Formally, we say that a label $m$ is \emph{successful on level $i$} if there is a
class $U$ on level $i-1$ and a class $U'$ on level $i$ such that $\gl(U)=\gl(U')=m$, and either $U'$
is the $F$-child of $U$, or $U'$ is not a child of $U$ at all.

\begin{example}
In Figure \ref{Fig:TGP}.(b), the only infinite branch
$\set{\zug{q,0}},\set{\zug{p,1}},\ldots$ is accepting.  At level $0$ this branch is
labeled with $0$. At each level $i>0$, we conservatively assume that the infinite branch beginning
with $\zug{q,0}$ goes through $\set{\zug{q,i}}$, and thus label $\set{\zug{q,i}}$ by $0$. As
$\set{\zug{q,i}}$ is proven finite on level $i+1$, we revise our assumption and continue to follow
the path through $\set{\zug{p,i}}$.  Since $\{\zug{p,i}\}$ is an $F$-class, the label $0$ is
successful on every level $i+1$.  Although the infinite branch is not labeled $0$ after the first
level, the label $0$ asymptotically approaches the infinite branch, checking along the way that the
branch is lexicographically minimal among the infinite branches through the root.
\end{example}

Theorem~\ref{Labeling_Succeeds} demonstrates that the global labeling captures the accepting or
rejecting nature of $\TGP$. Intuitively, at each level the class $U$ with label $m$ is on the
lexicographically minimal branch from $\first(m)$. If $U$ is on the lexicographically minimal
\emph{infinite} branch from $\first(m)$, the label $m$  is waiting for the branch to next reach an
$F$-class. If $U$ is not on the lexicographically minimal infinite  branch from $\first(m)$, then
$U$ is finite and $m$ is waiting for $U$ to die out. 

\begin{restatable}{theorem}{thmLabelingSucceeds}\label{Labeling_Succeeds}
A profile tree $\TGP$ is accepting iff there is a label $m$ that is successful infinitely often.
\end{restatable}
\begin{proof}
In one direction, assume there is a label $m$ that is successful
infinitely often. The label $m$ can be successful only when it occurs, and thus
$m$ occurs infinitely often, $\first(m)$ has infinitely many descendants, and
there is at least one infinite branch through $\first(m)$.  Let
$b=U_0,U_1,\ldots$ be the lexicographically minimal infinite branch that goes
through $\first(m)$.  We demonstrate that $b$ cannot have a suffix consisting
solely of non-$F$-classes, and therefore is an accepting branch.  By way of
contradiction, assume there is an index $j$ so that for every $k > j$, the
class $U_k$ is a non-$F$-class. By Lemma~\ref{Label_Props}.\ref{First_F},
$\first(m)$ is an $F$-class or the root and thus occurs before level $j$.

Let $\U=\set{W\mid W \prec_j U_j,~W\text{ is a descendant of $\first(m)$}}$ be
the set of descendants of $\first(m)$, on level $j$, that are lexicographically
smaller than $U_j$. Since $b$ is the lexicographically minimal infinite branch
through $\first(m)$, every class in $\U$ must be finite. Let $j' \geq j$ be the
level at which the last class in $\U$ dies out. At this point, $U_{j'}$ is the
lexicographically minimal descendant of $\first(m)$.  If $\gl(U_{j'}) \neq m$,
then there is no class on level $j'$ with label $m$, and, by Lemma
\ref{Label_Props}.\ref{Persistent_Labels}, $m$ would not occur after level $j'$.
Since $m$ occurs infinitely often, it must be that $\gl(U_{j'})=m$.  On every
level $k > j'$, the class $U_k$ is a non-$F$-child, and thus $U_k$
is the lexicographically minimal descendant of $U_{j'}$ on level $k$ and
so $\gl(U_k)=m$. This entails $m$ cannot be not successful after level $j'$,
and we have reached a contradiction.  Therefore, there is no such rejecting suffix of
$b$, and $b$ must be an accepting branch.

In the other direction, if there is an infinite accepting branch, then let
$b=U_0,U_1,\ldots$ be the lexicographically minimal infinite accepting branch.
Let $B'$ be the set of infinite branches that are lexicographically smaller
than $b$. Every branch in $B'$ must be rejecting, or $b$ would not be the
minimal infinite accepting branch.  Let $j$ be the first
index after which the last branch in $B'$ splits from $b$. Note that either
$j=0$, or $U_{j-1}$ is part of an infinite rejecting branch
$U_0,\ldots,U_{j-1},W_j,W_{j+1},\ldots$ smaller than $b$.  In both cases, we
show that $U_j$ is the first class for a new label $m$ that occurs on every
level $k>j$ of $\TGP$.

If $j=0$, then let $m=0$. As $m$ is the smallest label, and
there is a descendant of $U_j$ on every level of $\TGP$, it holds that $m$ will
occur on every level.  In the second case, where $j>0$, then $W_j$
must be the non-$F$-child of $U_{j-1}$, and so $U_j$ is the $F$-child.  Thus,
$U_j$ is given a new label $m$ where $U_j=\first(m)$.  For every label $m' < m$ and
level $k>j$, since for every descendant $U'$ of $U_j$ it holds that $W_k
\preceq_k U'$, it cannot be that $\lmd(m',k)$ is a descendant of $U_j$.  Thus,
on every level $k>j$, the lexicographically minimal descendant of $U_j$ will be
labeled $m$, and $m$ occurs on every level of $\TGP$.

We show that $m$ is successful infinitely often by defining an infinite
sequence of levels, $j_0,j_1,j_2,\ldots$ so that $m$ is successful on
$j_i$ for all $i>0$. As a base case, let $j_0=j$.  Inductively, at level $j_i$,
let $U'$ be the class on level $j_i$ labeled with $m$.  We have two cases. If
$U' \neq U_{j_i}$, then as all infinite branches smaller than $b$ have already
split from $b$, $U'$ must be finite in $\TGP$.  Let $j_{i+1}$ be the level at
which $U'$ dies out. At level $j_{i+1}$, $m$ will return to a descendant of
$U_{j_0}$, and $m$ will be successful. In the second case, $U'=U_{j_i}$.
Take the first $k>j_i$ so that $U_k$ is an $F$-class. As $b$ is an accepting
branch, such a $k$ must exist. As every class between $U_j$ and $U_k$ is a non-$F$-class,
$\gl(U_{k-1})=m$. If $U_k$ is the only child of $U_{k-1}$ then let $j_{i+1}=k$: since $\gl(U_k)=m$
and $U_k$ is not the non-$F$-child of $U_{k-1}$, it holds that $m$ is successful on level $k$.
Otherwise let $U'_k$ be the non-$F$-child of $U_{k-1}$, so that $\gl(U'_k)=m$. Again, $U'_k$ is
finite. Let $j_{i+1}$ be the level at which $U'_k$ dies out.  At level $j_{i+1}$, the label $m$ will
return to a descendant of $U_k$, and $m$ will be successful.
  \end{proof}

\subsection{Determining Lexicographically Minimal Descendants}
Recall that the definition of the labeling $\gl$ involves the computation of $\lmd(m,i)$, the class
with the minimal profile among all the descendants of $\first(m)$ on level $i$.  Finding $\lmd(m,i)$
requires knowing the descendants of $\first(m)$ on level $i$. We show how to store this information
with a partial order, denoted $\lreq{i}$, over classes that tracks which classes are minimal cousins of other classes.
Using this partial order, we can determine the class $\lmd(m,i+1)$ for every
label $m$ that occurs on level $i$, using only information about levels $i$ and
$i+1$ of $\TGP$. Lemma \ref{Label_Props}.\ref{Persistent_Labels} implies that
we can safely restrict ourselves to labels that occur on level $i$.

\begin{definition}\label{Def:Min_Cousin}
For two classes $U$ and $W$ on level $i$ of $\TGP$, say that $U$ is a \emph{minimal cousin} of $W$, written
$U \lreq{i} W$, iff $W$ is a descendant of $\first(\gl(U))$.  Say $U \lr{i} W$ when $U \lreq{i} W$
and $U \neq W$.
\end{definition}

For a label $m$ and level $i$, we can determine $\lmd(m,i+1)$ given only the classes on levels $i$
and $i+1$ and the partial order $\lr{i}$. Let $U$ be a class $U$ on level $i$.  Because labels can
move between branches, the minimal descendant of $\first(\gl(U))$ on level $i+1$ may be a nephew of
$U$, not necessarily a direct descendant.  Define the $\lreq{i}$-nephew of $U$ as
$\lsf_i(U)=\min_{\preceq_{i+1}}(\set{W'\mid W\text{ is the parent of }W'\text{ and } U \lreq{i}
W})$.

\begin{restatable}{lemma}{lemLRProvidesLMD}\label{LR_Provides_LMD}
For a class $U$ on level $i$ of $\TGP$, it holds that $\lmd(\gl(U),i+1) = \lsf_i(U)$.
\end{restatable}
\begin{proof}
We prove that $\set{W' \mid W \text{ is the parent of }W'\text{ and } U \lreq{i} W}$ contains every
descendant of $\first(\gl(U))$ on level $i+1$, and thus that its minimal element 
is $\lmd(\gl(U),i+1)$.  Let $W'$ be a class on level $i+1$, with parent $W$ on level
$i$.  If $U \lreq{i} W$, then $W$ is a descendant of $\first(\gl(U))$ and $W'$ is likewise a
descendant of $\first(\gl(U))$.  Conversely, as $\gl(U)$ exists on level $i$, if $W'$ is a descendant of
$\first(\gl(U))$, then its parent $W$ must also be a descendant of $\first(\gl(U))$ and $U \lreq{i}
W$. 
\end{proof}

By using $\lsf_i$, we can in turn define the set of valid labels for a class $U'$ on level $i+1$.
Formally, define the $\lreq{i}$-uncles of $U'$ as $\lpf_i(U')= \set{U \mid U'= \lsf_i(U)}$.  Lemma
\ref{LPF_Determines_Labels} demonstrates how $\lpf_i$ corresponds to $\labelsf$.

\begin{restatable}{lemma}{lemLPFDeterminesLabels}\label{LPF_Determines_Labels}
Consider a class $U'$ on level $i+1$. The following hold:
\begin{compactenum}
\item $\labelsf(U') \cap \set{\gl(W) \mid W \text{ on level $i$}}  = \set{\gl(U) \mid U \in \lpf_i(U')}$.
\item $\labelsf(U')=\emptyset$ iff $\lpf_i(U')=\emptyset$.
\end{compactenum}
\end{restatable}
\begin{proof}\ 
\begin{compactenum}
\item Let $U$ be a class on level $i$.  By definition, $\gl(U) \in \labelsf(U')$ iff
$U'=\lmd(\gl(U),i+1)$.  By Lemma~\ref{LR_Provides_LMD}, it holds that $\lmd(\gl(U),i+1)=\lsf_i(U)$.
By the definition of $\lpf_i$,  we have that $U'= \lsf_i(U)$ iff $U \in \lpf_i(U')$.  Thus every
label in $\labelsf(U')$ that occurs on level $i$ labels some node in $\lpf_i(U')$.
\item If $\lpf_i(U') \neq \emptyset$, then part (1) implies $\labelsf(U') \neq \emptyset$. In other
direction, let $m=\min(\labelsf(U'))$. By Lemma \ref{Label_Props}.\ref{Persistent_Labels}, there is a $U$
on level $i$ so that $\gl(U)=m$, and by part (1) $U \in \lpf_i(U')$. 
\end{compactenum}
\end{proof}

Finally, we demonstrate how to compute $\lreq{i+1}$ only using information about the level $i$ of
$\TGP$ and the labeling for level $i+1$. As the labeling depends only on $\lreq{i}$, this removes
the final piece of global information used in defining $\gl$.

\begin{restatable}{lemma}{lemLMDSuccessor}\label{LMD_Successor}
Let $U'$ and $W'$ be two classes on level $i+1$ of $\TGP$, where $U' \neq W'$. Let $W$ be the parent
of $W'$. We have that $U' \lreq{i+1} W'$ iff there exists a class $U$ on level $i$ so that
$\gl(U) = \gl(U')$ and $U \lreq{i} W$. 
\end{restatable} 
\begin{proof}
If there is no class $U$ on level $i$ so that $\gl(U) = \gl(U')$, then $U'=\first(\gl(U'))$. Since
$W'$ is not a descendant of $U'$, it cannot be that $U' \lreq{i+1} W'$.  If such a class $U$ exists,
then $U \lreq{i} W$ iff $W$ is a descendant of $\first(\gl(U))$, which is true iff $W'$ is a
descendant of $\first(\gl(U'))$: the definition of $U' \lreq{i+1} W'$.
\end{proof}

\subsection{Reusing Labels}\label{Sect:Rabin}

As defined, the labeling function $\gl$ uses an unbounded number of labels. However, as there are at
most $\abs{Q}$ classes on a level, there are at most $\abs{Q}$ labels in use on a level. We can thus use a
fixed set of labels by reusing dead labels. For convenience, we use $2\abs{Q}$ labels, so that we
never need reuse a label that was in use on the previous level. \fullv{We demonstrate how to use $\abs{Q}-1$
labels in Appendix~\ref{App:Constructions}.}{The full version demonstrates how to use $\abs{Q}-1$
labels.}  There are two barriers to reusing labelings.  First, we
can no longer take the numerically minimal element of $\labelsf(U)$ as the label of $U$. Instead, we
calculate which label is the oldest through $\preceq$.  Second, we must ensure that a label that is
good infinitely often is not reused infinitely often. To do this, we introduce a Rabin condition to
reset each label before we reuse it.

We inductively define a sequence of labelings, $l_i$, each from the $i$th level of $\TGP$ to
$\{0, \ldots, 2\abs{Q}\}$.  As a base case, there is only one equivalence class $U$ on level $0$ of
$\TGP$\!, and define $l_0(U)=0$.  Inductively, given the set of classes ${\U}_i$ on level $i$, the
function $l_i$, and the set of classes ${\U}_{i+1}$ on level $i+1$, we define $l_{i+1}$ as follows.
Define the set of unused labels $\mathrm{FL}(l_i)$ to be $\set{m\mid m\text{ is not in the range
of $l_i$}}$. As $\TGP$ has bounded width $\abs{Q}$, we have that  $\abs{Q} \leq
\abs{\mathrm{FL}(l_i)}$. Let $\fl_{i+1}$ be the $\zug{\preceq_{i+1}, <}$-minjection from $\set{U'
\text { on level i+1} \mid \lpf_i(U') =\emptyset}$ to $\mathrm{FL}(l_i)$.  Finally, define the
labeling $l_{i+1}$ as

$$l_{i+1}(U') = 
\begin{cases}
l_i(\min_{\preceq_i}(\lpf_i(U'))) & \text {if } \lpf_i(U') \neq \emptyset,\\
\fl_{i+1}(U') & \text{if }\lpf_i(U') = \emptyset.\\
\end{cases}
$$

Because we are reusing labels, we need to ensure that a label that is good infinitely often is not
reused infinitely often. Say that a label $m$ is \emph{bad in $l_i$} if $m \not \in
\mathrm{FL}(l_{i-1})$, but $m \in \mathrm{FL}(l_i)$.  We say that a label $m$ is \emph{good in
$l_i$} if there is a class $U$ on level $i-1$ and a class $U'$ on level $i$ such that
$l_{i-1}(U)=l_i(U')=m$ and $U'$ is either the $F$-child of $U$ or is not a child of $U$ at all.

Theorem~\ref{Rabin_Labeling_Succeeds} demonstrates that the Rabin condition of a label being good
infinitely often, but bad only finitely often, is a necessary and sufficient condition to $\TGP$
being accepting. The proof\fullv{, given in Appendix~\ref{App:Proofs},}{, ommitted for brevity,} associates each label $m$ in $\gl$
with the label $l_i(\first(m))$. 

\begin{restatable}{theorem}{thmRabinLabelingSucceeds}\label{Rabin_Labeling_Succeeds}
A profile tree $\TGP$ is accepting iff there is a label $m$ where $\set{i \mid m\text{ is bad in
}l_i}$ is finite, and $\set{i \mid m\text{ is good in }l_i}$ is infinite.
\end{restatable}

\section{A New Determinization Construction for \buchi Automata}\label{Sect:Definition}

In this section we present a determinization construction for $\A$ based on the profile tree $\TGP$.
For clarity, we call the states of our deterministic automaton \emph{macrostates}.
\begin{definition}
Macrostates over $\A$ are six-tuples $\zug{S,\preceq,l,\lreqd,G,B}$ where:
\begin{compactitem}
\item $S\subseteq Q$ is a set of states.
\item $\preceq$ is a linear preorder over $S$.
\item $l \colon S \to \set{0,\ldots, 2\abs{Q}}$ is a labeling.
\item ${\lreqd} \subseteq {\preceq}$ is another preorder over $S$.
\item $G, B$ are sets of good and bad labels used for the Rabin condition.
\end{compactitem}
\end{definition}
For two states $q$ and $r$ in $Q$, we say that $q \approx r$ if $q \preceq
r$ and $r \preceq q$.  We constrain the labeling $l$ so that it characterizes the equivalence
classes of $S$ under $\preceq$, and the preorder $\lreqd$ to be a partial order over the equivalence
classes of $\preceq$.  Let $\bQ$ be the set of macrostates.  

\begin{figure}[tb]
\begin{centering}
\subfloat[An automaton $\B$]{
\raisebox{0.3in}{
\begin{tikzpicture} [->,auto,node distance=0.6cm,line width=0.4mm]
\node[state,initial,inner sep=1pt] (q) {$q$};
\node[state,accepting] (p) [right=of q] {$p$};
\path 	(q) edge [loop above] node {a} (q)
			edge [bend left] node {a} (p)
		(p) edge [bend left] node {b} (q)
			edge [loop above] node {a,b} (p);
\end{tikzpicture}
}}
~~~~~~~~~\subfloat[The first four macrostates in the run of $D^R(\B)$ on $ab^\omega$.]
{
\begin{tikzpicture} [auto, line width=0.3mm, node distance=0.6cm] 
\node at (0,2) [anchor=west,rectangle, draw, label=left:$\q_0~\equal$] (p0) 
             {${\zug{\set{q}^0},~\emptyset,~G=\emptyset,~B=\emptyset}$};
\node at (0,1) [anchor=west,rectangle, draw, label=left:$\q_1~\equal$] (p1) 
             {${\zug{\set{q}^0 \prec \set{p}^1},~q \lrb p,~G=\emptyset,~B=\emptyset}$};
\node at (0,0) [anchor=west,rectangle, draw, label=left:$\q_2~\equal$] (p2) 
             {${\zug{\set{q}^0 \prec \set{p}^2},~q \lrb p,~G=\set{0},~B=\set{1}}$};
\node at (0,-1) [anchor=west,rectangle, draw, label=left:$\q_3~\equal$] (p3) 
             {${\zug{\set{q}^0 \prec \set{p}^1},~q \lrb p,~G=\set{0},~B=\set{2}}$};
\end{tikzpicture}
}
\caption{An automaton and four macrostates.  For each macrostate $\zug{S,\preceq,l,\lreqd,G,B}$,
we first display the equivalence classes of $S$ under $\preceq$ in angle brackets, superscripted
with the labels of $l$. We then display the $\lreqd$ relation, and finally the sets $G$ and $B$.
\label{Fig:Macrostates}}
\end{centering}
\end{figure}
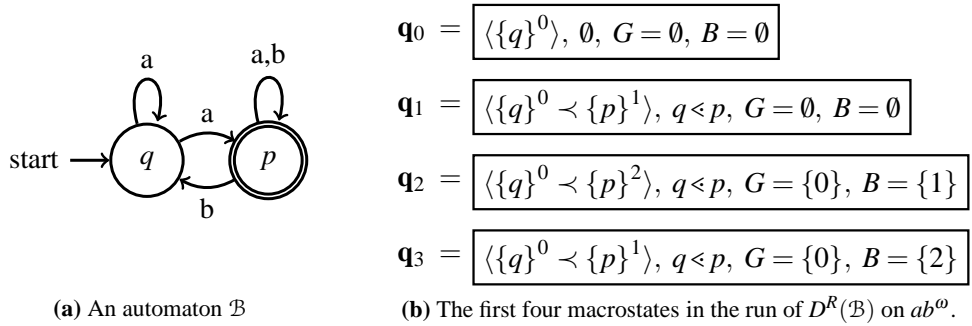

\pagebreak[3]
Before defining transitions between macrostates, we reproduce the pruning of edges
from $\Gprime$ by restricting the transition function $\rho$ with respect to $S$ and
$\preceq$.  For a state $q \in S$ and $\sigma \in \Sigma$, let $\rhoprec(q,\sigma)= \set{q' \in
\rho(q,\sigma)\mid \text{for every }r \in \rho^{-1}(q',\sigma) \cap S,~ r \preceq q}$.  Thus, when a
state has multiple incoming $\sigma$-transitions from $S$, the function $\rhoprec$ keeps only the
transitions from states maximal under the $\preceq$ relation.  For every state $q' \in \rho(S,\sigma)$, the
set $\rhoprec^{-1}(q',\sigma) \cap S$ is an equivalence class under $\preceq$.  We note that
$\rho(S,\sigma)=\rhoprec(S,\sigma)$. 

\begin{example}
Figure~\ref{Fig:Macrostates} displays the first four macrostates in a run of this determinization
construction. 
Consider the state $\q_1=\zug{\set{q,p},\preceq,l,\lreqd,\emptyset,\emptyset}$ where $q \prec p$, $q
\lreqb p$, $l(q) = 0$, and $l(p)=1$.  We have $\rho(q,a)~=~\set{p,q}$.  However, $p \in \rho(p,a)$
and $q \prec p$.  Thus we discard the transition from $q$ to $p$, and $\rhoprec(q,a) = \set{q}$.  In
contrast, $\rhoprec(p,a) = \rho(p,a) = \set{p}$, because while $p \in \rho(q,a)$, it holds that $q
\prec p$.
\end{example}

For $\sigma \in \Sigma$, we define the $\sigma$-successor of $\zug{S,\preceq,l,\lreqd,G,B}$ to be
$\zug{S',\preceq',l',\lreqd',G',B'}$ as follows.  First, $S' = \rho(S,\sigma)$.
Second, define $\preceq'$ as follows. For states $q', r' \in S'$, let $q \in
\rhoprec^{-1}(q',\sigma)$ and $r \in \rhoprec^{-1}(r',\sigma)$. As the parents
of $q'$ and $r'$ under $\rhoprec$ are equivalence classes the choice of $q$ and
$r$ is arbitrary.
\begin{compactitem}
\item If $q \prec r$, then $q' \prec' r'$.
\item If $q \approx r$ and $q' \in F$ iff $r' \in F$, then $q' \approx' r'$.
\item If $q \approx r$, $q' \not\in F$, and $r' \in F$, then $q' \prec' r'$.
\end{compactitem}

\begin{example}
As a running example we detail the transition from
$\q_1=\zug{\set{q,p},\preceq,l,\lreqd,\emptyset,\emptyset}$ to\\
$\q_2=\zug{S',\preceq',l',\lreqd',G',B'}$ on $b$.  We have $S' = \rho(\set{q,p},b) = \set{q,p}$.  To
determine $\preceq'$, we note that $p \in S$ is the parent of both $q \in S'$ and $p \in S'$.  Since
$q \not \in F$, and $p \in F$, we have $q \prec' p$.
\end{example}

Third, we define the labeling $l'$ as follows.
As in the profile tree $T$, on each level we give the label $m$ to the minimal
descendants, under the $\preceq$ relation, of the first equivalence class to be
labeled $m$.  For a state $q \in S$, define the \emph{nephews of $q$} to be
$\neph(q,\sigma) = \min_{\preceq'}(\rhoprec(\set{r \in S \mid q \lreqb r},
\sigma))$.  Conversely, for a state $r' \in S'$ we define the \emph{uncles of
$r'$} to be be $\unc(r',\sigma)=\set{q \mid r' \in
\neph(q,\sigma)}$.

Each state $r' \in S'$ inherits the oldest label from its uncles. If $r'$ has no uncles, it gets a
fresh label.  Let $\mathrm{FL}(l)= \set{m\mid m\text{ not in the range of $l$}}$ be the free labels
in $l$,  and let $\fl$ be the $\zug{\preceq',<}$-minjection from $\{r' \in S'\mid \unc(r',\sigma)
= \emptyset\}$ to $\mathrm{FL}(l)$, where $<$ is the standard order on $\set{0,\ldots,2\abs{Q}}$.
Let
$$l'(r') = \begin{cases} 
l(q),~\text{ for some }q \in \min_\preceq(\unc(r',\sigma)) & \text {if } \unc(r',\sigma) \neq \emptyset,\\
\fl(r') & \text{if }\unc(r',\sigma)
= \emptyset.\\ \end{cases} $$

\begin{example}
The nephews of $q \in S$ is the $\preceq'$-minimal subset of the set 
$\rhoprec(\set{r \in S \mid q \lreqb r}, \sigma)$.  Since $q \lreqb q$ and $q \lreqb p$, we have
that $\neph{q,b}=\min_{\preceq'}(\set{q,p})=\set{q}$. Similarly, for $p \in S$ we have $p \lreqb p$ and
$\neph(p,b)=\min_{\preceq'}(\set{p,q})=\set{q}$.  Thus for $q \in S'$, we have
$\min_\preceq(\unc(q,b))=\min_\preceq(\set{p,q})=\set{q}$ and we set $l'(q)=l(q)=0$.
For $p \in S'$, we have $\unc(p,b)=\emptyset$ and $l'(p)$ is the first unused label:
$l'(p)=2$.
\end{example}

Fourth, define the preorder $\lreqd'$ as follows.  For states $q', r' \in S'$, define $q' \lreqbp r'$ iff
$q' \approx' r'$ or there exist $q, r \in S$ so that: $r' \in \rhoprec(r, \sigma)$;  $q \in
\unc(q',\sigma)$;  and $q \lreqb r$. The labeling $l'$ depends on recalling which states descend
from the first equivalence class with a given label, and  $\lreqd'$ tracks these descendants.  

Finally, for a label $m$ let $S_m = \set{r \in S \mid l(r)=m}$ and $S'_m = \set{r' \in S' \mid
l'(r')=m}$ be the states in $S$, resp $S'$, labeled with $m$.  Recall that a label $m$ is good
either when the branch it is following visits $F$-states, or the branch dies and it moves to another
branch. Thus say $m$ is \emph{good} when: $S_m \neq \emptyset$; $S'_m \neq \emptyset$; and either
$S'_m \subseteq F$ or $\rhoprec(S_m,\sigma) \cap S'_m = \emptyset$.  $G'$ is then $\set{m \mid
m\text{ is good}}$.  Conversely, a label is bad when it occurs in $S$, but not in $S'$.  Thus the
set of \emph{bad} labels is $B'=\set{m \mid S_m \neq \emptyset,~S'_m = \emptyset}$.

\begin{example}
As $p \in \rhoprec(p,b)$; $q \in \unc(q,b)$; and $q \lrb p$, we have $q \lrbp p$.
Since $l(q)=0$ and $l'(q)=0$, but $q \not \in \rhoprec(q,b)$, we have $0 \in G'$, and as
nothing is labeled $1$ in $l'$, we have $1 \in B'$. 
\end{example}

Lemma~\ref{Sigma_Successors_Valid}\fullv{, proven in Appendix~\ref{App:Proofs}}{, proven in the full version,} states that
$\zug{S',\preceq',l',\lreqd',G',B'}$ is a valid macrostate.

\begin{restatable}{lemma}{lemSigmaSuccessorsValid}\label{Sigma_Successors_Valid}
For a macrostate $\q \in \bQ$ and $\sigma \in \Sigma$, the $\sigma$-successor of $\q$ is a macrostate.
\end{restatable}

\begin{definition}\label{Def:Rabin_Condition}
Define the DRW automaton $D^R(\A)$ to be $\zug{\Sigma,\bQ, \bQ^{in},
{\bf \rho_Q},\alpha}$, where:
\begin{compactitem}
\item $\bQ^{in}=\set{\zug{Q^{in},\preceq_0,l_0, \lreq{0}, \emptyset, \emptyset}}$, where:
\begin{compactitem} 
\item $\mathord{\preceq_0}=\lreq{0}=Q^{in} \times Q^{in}$
\item $l_0(q)=0$ for all $q \in Q^{in}$
\end{compactitem}
\item For $\q \in \bQ$ and $\sigma \in \Sigma$, let ${\bf \rho_Q}(\q,\sigma)= \set{\qp}$, where
$\qp$ is the $\sigma$-successor of $\q$
\item 
$\alpha=  \zug{G_0,B_0},\ldots,\zug{G_{2\abs{Q}},B_{2\abs{Q}}}$,
where for a label $m \in \set{0, \ldots, 2\abs{Q}}$:
\begin{compactitem}
\item $G_m = \set{\zug{S,\preceq,l,\lreqd,G,B} \mid m \in G}$
\item $B_m = \set{\zug{S,\preceq,l,\lreqd,G,B} \mid m \in B}$
\end{compactitem}
\end{compactitem}
\end{definition}

Theorem~\ref{Construction_Correct}, proven in \fullv{Appendix~\ref{App:Proofs}}{the full version},
asserts the correctness of the construction and says that its blowup is
comparable with known determinization constructions. \hide{The connection
between the construction and the tree of equivalence classes is formalized, and
the Theorem is proven, in Appendix~\ref{App:Proofs}.}

\begin{restatable}{theorem}{thmConstructionCorrect}\label{Construction_Correct}
For an NBW $\A$ with $n$ states, $L(D^R(\A)) = L(\A)$
and $D^R(\A)$ has $n^{O(n)}$ states.
\end{restatable}

There are two simple improvements to the new construction, detailed in
\fullv{Appendix~\ref{App:Constructions}}{the full version}.  First, we do not need $2\abs{Q}$
labels: it is sufficient to use $\abs{Q}-1$ labels. Second, Piterman's technique of dynamic renaming
can reduce the Rabin condition to a parity condition.

\section{Discussion}

In this paper we extended the notion of profiles from \cite{FKVW11} and developed a theory of profile trees. This theory affords a novel determinization construction, where determinized-automaton states are sets of input-automaton states augmented with two preorders.  In the future, a more thorough analysis could likely improve the upper bound on the size of our construction.  We hope to see heuristic optimization techniques developed for this construction, just as heuristic optimization techniques were developed for Safra's construction \cite{TFVT10}.

More significantly, profile trees afford us the first theoretical underpinnings for determinization.
Decades of research on \buchi determinization have resulted in a plethora of constructions, but a
paucity of mathematical structures underlying their correctness.  This is the first new major line
of research in \buchi determinization since \cite{MS95}, and we expect it to lead to further research in this important area.

One important question is to understand better the connection between profile trees and Safra's construction.  A key step in the transition between Safra trees is to remove states if they appear in more than one node. This seems analogous to the pruning of edges from $\Gprime$.  The second preorder in our construction, namely the relation $\lreq{i}$, seems to encodes the order information embedded in Safra trees.  Perhaps our approach could lead to declarative definition of constructions based on Safra and Muller-Schupp trees.  In any case, it is our hope that profile trees will encourage the development of new methods to analyze and optimize determinization constructions.

\bibliographystyle{eptcs}
\bibliography{gandalf13}

\fullv{}{\end{document}}
\newpage
\appendix

\section{Proofs}\label{App:Proofs}

\hide{
\subsection{Section 3}
\lemGprimeParentsEquivalent*

\lemTGPBinTree*

\thmTGPCaptures*

}
\subsection{Lemma \ref{Label_Props}}

\lemLabelProps*
\begin{proof}
Parts \ref{L_Distinct} through \ref{L_Preceq}  follow immediately from the fact that
$U=\lmd(\gl(U),i)$. Part \ref{First_F} follows from the fact that, for every class $W$ on level $i$
with non-$F$-child $W'$, we have $W'=\lmd(\gl(U),i+1)$. Part \ref{Less_Earlier} follows from the
definition of $\labelsf$: a new label is always larger than every earlier label. Finally, we prove
part \ref{Persistent_Labels}.

Assume $U'$, on level $i$, is such that $\gl(U')=m$ and $U' \neq \first(m)$. By Part
\ref{L_Descends}, there must be a descendant of $\first(m)$ on level $i-1$. Let $U=\lmd(m,i\sminus
1)$. To prove $\gl(U)=m$, we show $m=\min(\labelsf(U))$.  Consider $m' < m$ such that $U$ is a
descendant of $\first(m')$, and thus $U'$ is also descendant of $\first(m')$.  By Part
\ref{Less_Earlier}, $\first(m')$ occurs before $\first(m)$.  Since $U$ is a descendant of both
$\first(m)$ and $\first(m')$, it must be that $\first(m)$ is a descendant of $\first(m')$. 

Since $m' < m$, if $U'=\lmd(m',i)$ then $\gl(U')$ would be $m'$. There must then exist a $W' \prec_i
U'$ that is a descendant of $\first(m')$, but not a descendant of $\first(m)$. By the definition of
lexicographic order, $W'$ is lexicographically smaller than every descendant, on level $i$,  of
$\first(m)$. Let $W$ be the parent of $W'$. We have that $W$ is a descendant of $m'$ that is
lexicographically smaller than every descendant, on level $i-1$, of $\first(m)$. In specific, $W
\prec_{i-1} U$, and thus $U \neq \lmd(m',i)$.  Thus $m=\min(\labelsf(U))=\gl(U)$. 
\end{proof}


\hide{
\subsection{Lemmas \ref{LR_Provides_LMD}, \ref{LPF_Determines_Labels}, and \ref{LMD_Successor}}
\lemLRProvidesLMD*

\lemLPFDeterminesLabels*

\lemLMDSuccessor*
}

\subsection{Theorem \ref{Rabin_Labeling_Succeeds}}
To show a correlation between the labeling in Section \ref{Sect:Labeling} and the labeling here, we
define a mapping, $f$, from the labels of $l$ to $\set{0, \ldots,
2\abs{Q}}$. For a label $m$, where $\first(m)$ occurs on level $i$, let $f(m)= l_i(\first(m))$.

\begin{lemma}\label{Rabin_Corresponds}
For classes $U$ on level $i$ and $U'$ on level $i+1$, if $\gl(U) = \gl(U')$, then
$l_i(U)=l_{i+1}(U')=f(\gl(U))$.
\end{lemma}
\begin{proof}
Let $k$ be the number of levels between $U$ and $\first(\gl(U))$. We prove this lemma by induction
over $k$.  As a base case, if $k=0$, then $U=\first(\gl(U))$ and by definition $f(\gl(U))=l_i(U)$.
Inductively, assume $k >0$, and assume this lemma holds for every $W$ at most $k-1$ steps removed from
$\first(\gl(W))$. Since $k > 0$, then $U \neq \first(\gl(U))$. Let $W$ be the node on level $i-1$ such
that $\gl(W)=\gl(U)$. By the inductive hypothesis, $l_{i-1}(W)=l_i(U)$. Further, since $\first(\gl(W))
= \first(\gl(U))$, we have $l_i(U)=f(\gl(U))$.  We now show that $U =
\min_{\preceq_i}(\lpf_i(U'))$.

As $\gl(U) = \gl(U')$, we have that $\gl(U) \in
\labelsf(U')$. By Lemma \ref{LPF_Determines_Labels}, this implies $U \in \lpf_i(U')$. To prove that
$U=\min_{\preceq_i}(\lpf_i(U'))$, let $W \in \lpf_i(U')$ be another class on level $i$. By
Lemma \ref{LPF_Determines_Labels}, this implies $\gl(W) \in \labelsf(U')$, and thus $\gl(U) <
\gl(W)$.  As $U'$ is a descendant of both $\first(\gl(U))$ and $\first(\gl(W))$, one is a descendant of
the other.  Since $\gl(U) < \gl(W)$, by Lemma \ref{Label_Props}.\ref{Less_Earlier} it must be that $\first(\gl(W))$  is a descendant of $\first(\gl(U))$.
Thus $W$ is a descendant of $\first(\gl(U))$, and by Lemma \ref{Label_Props}.\ref{L_Preceq} we have
$U \preceq W$. Therefore  $U=\min_{\preceq_i}(\lpf_i(U'))$, and $l_{i+1}(U') = l_i(U)$.
\end{proof}

\begin{corollary}\label{LPF_Corollary}
For every class $U$ on level $i$, it holds that $l_i(U)=f(\gl(U))$.
\end{corollary}

\thmRabinLabelingSucceeds*
\begin{proof}
We prove a relation with Theorem \ref{Labeling_Succeeds}. For the first direction, let $m$ be a label that
is successful infinitely often. We prove that $f(m)$ is bad in only finitely many $l_i$, and is good
in infinitely many $l_j$. Let $U$ on level $j$ be $\first(m)$. First, as $m$ occurs on every level
$k > j$, Lemma \ref{Rabin_Corresponds} implies $f(m)$ occurs on $k$, and thus $f(m)$ is not bad in
$l_k$.  Second, let $k > j$ be a level on which $m$ is successful. This implies there exist
classes $U$ on level $k-1$ and $U'$ on level $k$, so that $\gl(U)=\gl(U')=m$ and $U'$ is not
the non-$F$-child of $U$. Lemma \ref{Rabin_Corresponds} implies that $l_{k-1}(U)=l_{k}(U')=f(m)$, and
thus that $f(m)$ is good in $l_k$. We thus conclude $f(m)$ is good in infinitely many $l_k$. 

For the other direction, let $m'$ be a label that is bad in $l_i$ for finitely many $i$, and is good
in $l_i$ for infinitely many $i$.  Since $m'$ is bad only finitely often, there is some level after
which $m'$ is not bad.  Let $j$ be the first level after which $m'$ ceases being bad on which $m'$
occurs. This implies $m'$ occurs on every level $k > j$. Let $U$ on level $j$ be such that
$l_j(U)=m'$. Since $m'$ does not occur on $j-1$, it must be that $\lpf_j(U) = \emptyset$: otherwise
$l_j(U)$ would be $l_j(\min_{\preceq_j}(\lpf_j(U')))$.  Thus by Lemma~\ref{LPF_Determines_Labels} we
have that $\labelsf(U)=\emptyset$, and there is a label $m$ in $l$ so that $U=\first(m)$, and
$f(m)=m'$. By assumption, there are infinitely many $k > j$ so that $m'$ succeeds in $l_k$. On each
of these $k$'s, there is a class $U$ on level $k-1$ and $U'$ on level $k$ so that $l_{k-1}(U)=m'$,
$l_k(U')=m'$, and $U'$ is not the non-$F$-child of $U$. By Corollary \ref{LPF_Corollary}, $m=\gl(U)$
is good on level $k$, and $m$ is good infinitely often.
\end{proof}

\subsection{Connecting $\TGP$ to $D^R(\A)$.}

In this Appendix we prove the machinery of $D^R(\A)$ matches the inductive definitions of labeling
over $\TGP$. We first prove the the transitions of $D^R(\A)$ are valid.

\lemSigmaSuccessorsValid*
\begin{proof}
As $\zug{S,\preceq,l,\lreqd,G,B}$ is a macrostate, we have $\preceq$ is a linear preorder,
$\lreqb \subseteq \preceq$, and for every $q, r, s, t \in S$: $q \approx r$ iff $l(q) = l(r)$; $q
\approx r$ iff $q \lreqb r$ and $r \lreqb q$; and if $q \approx r$, $s \approx t$, and $q \lreqb s$,
then $r \lreqb t$. We must prove this also holds for $\lreqd'$, $\preceq'$, and $l'$ over states in
$S'$. Below, let $q', r', s', t'$ be states in $S'$, and $q, r, s, t \in S$ be such that $q' \in
\rhoprec(q, \sigma)$, $r' \in \rhoprec(r, \sigma)$, $s' \in \rhoprec(s,\sigma)$, and $t'
\in \rhoprec(t,\sigma)$.
\hide{
We restate the definitions of $\zug{S',\preceq',l',\lreqd',G,B}$. 
\begin{enumerate}
\item $\rhoprec(q,\sigma) = \set{q' \in \rho(q,\sigma)\mid \text{for every $r \in S$, if $q' \in
\rho(r,\sigma)$ then }r \preceq q}$.  Recall that the set of states $\set{q \in S \mid q' \in
\rhoprec(q,\sigma)}$ are an equivalence class under $\preceq$.
\item $\preceq'$ is defined so that  
\begin{compactitem}
\item If $q \prec r$, then $q' \prec' r'$
\item If $q \approx r$ and $q', r' \in F$, then $q' \approx' r'$.
\item If $q \approx r$ and $q' r' \not\in F$  then $q' \approx' r'$.
\item If $q \approx r$, $q' \not\in F$, and $r' \in F$, then $q' \prec' r'$. 
\end{compactitem}
\item $l'$ is defined as follows:
\begin{compactenum}
\item $\neph(q,\sigma) = \min_{\preceq'}\set{r' \mid \text{exists }r \in S,~q \lreqb r,~r' \in \rhoprec(r,\sigma)}$.
\item $\unc(r',\sigma)=\set{q \mid r' \in \neph(q,\sigma)}$
\item $\mathrm{FL}(l)= \set{m\mid m\text{ not in the range of $l$}}$
\item $\fl$ is the $\zug{\preceq',<}$-minjection from $\set{s' \in S' \mid \unc(s',\sigma) =
\emptyset}$ to $\mathrm{FL}(l)$
\item $l'(r') =
\begin{cases} 
l(q),~q \in \min_{\preceq}(\unc(r',\sigma)) & \text {if } \unc(r',\sigma) \neq \emptyset,\\
\fl(r') & \text{if }\unc(r',\sigma) = \emptyset.\\
\end{cases} $
\end{compactenum}
\item $q' \lreqbp r'$ iff $q' \approx' r'$ or there exists $q_2 \in \unc(q')$, so that $l(q_2) =
l(q')$ and $q_2 \lreqb r$.
\end{enumerate}
}

To demonstrate that $\preceq'$ is a linear preorder, we show it is reflexive, relates every two
elements, and is transitive. That $\preceq'$ is reflexive follows from the definition.  To show that
$\preceq'$ relates every two elements, note that as $\preceq$ is a linear preorder, either $q \prec
r$, $r \prec q$, or $q \approx r$. By the definition of $\preceq'$, either $q' \prec' r'$, $q'
\approx' r'$, or $r' \prec' q$.  To show that $\preceq'$ is transitive, assume $q' \preceq' r'
\preceq' s'$. By definition of $\preceq'$ we then have $q \preceq r$ and $r \preceq s$. Since
$\preceq$ is transitive, we have $q \preceq s$.  In order for $q' \not\preceq' s'$, it would need to
be that $q \approx s$, $q' \in F$, and $s' \not\in F$. If $q \approx s$, then $q \approx r$ and $r
\approx s$.  Thus if $r' \in F$, we would have $s' \preceq r'$, a contradiction.  If $r' \not \in
F$, we would have $r' \preceq q'$, a contradiction. Thus it cannot be the case that $q \approx s$,
$q' \in F$, and $s' \not \in F$, and either $q' \approx' s'$, or $q' \prec' s'$, and $\preceq'$ is
transitive and a linear preorder.

Next, we prove the labeling must give unique labels to the equivalence classes of $S'$ under
$\preceq$': that $q' \approx' r'$ iff $l'(q')=l'(r')$.  By the above properties, if $q
\approx r$, then $\neph(q,\sigma)=\neph(r,\sigma)$.  Further, $\neph(q,\sigma)$ is
an equivalence class under $\preceq'$ or is empty.  In one direction, let $q' \approx' r'$ and let
$m=l'(q')$. That $q' \approx r'$ implies for every $q \in \unc(q',\sigma)$ that $q' \approx' r'$,
and thus $\unc(q',\sigma)=\unc(r',\sigma)$. If $\unc(q',\sigma)=\emptyset$, then
$\unc(r',\sigma)=\emptyset$. As a minjection maps equivalent elements to the same value, we
have  $l(q')=\fl(q')=\fl(r')=l(r')$. Alternately, if $\unc(q',\sigma) \neq \emptyset$ then $m=l(q)$
for $q \in \unc(q',\sigma)$, and $q \in \unc(r',\sigma)$, and $l(r')=m$. 

Finally, we must demonstrate two things about $\lreqd'$: that $\lreqd' \subseteq \preceq'$, and that
$\lreqd'$ is a partial order over the equivalence classes of $\preceq'$. Assume $q' \lreqbp r'$. If
$q' \approx' r'$, then $q' \preceq' r'$. Otherwise there exists a $q_2 \in \unc(q')$ so that
$l(q_2)=l(q')$ and $q_2 \lreqb r$. This implies both $r', q' \in \neph(q_2,\sigma)$. Since
$l(q')=l(q_2)$, it must be that $q_2 \in \unc(q', \sigma)$ and thus $q' \in
\neph(q_2,\sigma)$. Thus $q' \preceq' r'$, and $\lreqd' \subseteq \preceq'$.
This implies $q' \approx' r'$ iff $q' \lreqbp r'$ and $r' \lreqbp q'$ It remains to show if $q'
\approx' r'$, $s' \approx t'$, and $q' \lreqbp s'$, then $r' \lreqbp t'$. If $q' \approx' s'$, then
$r' \approx' t'$ and $r' \lreqbp t'$. Otherwise there exists a $q_2$ so that $l(q_2)=l(q')$ and $q_2
\lreqb s$. Since $r' \approx' q'$, it holds that $l(r')=l(q_2)$. Since $s' \approx' t'$, it holds
that $s \approx t$ and $q_2 \lreqb t$.  Thus $r' \lreqbp t'$, and we have satisfied all requirements
for $\zug{S',\preceq',l',\lreqd',G,B}$ to be a macrostate.
\end{proof}

Lemma \ref{Gprime_Match} asserts that the set of states $S$ and the preorder $\preceq$
correspond to the nodes on a level $i$ of $\Gprime$ and the preorder $\preceq_i$.  Further, the
edges in $\Gprime$ correspond to transitions in $\rhoprec$. The proof relates $\sigma$-successors of
macrostates and $\preceq_{i}$.

\begin{lemma}\label{Gprime_Match}
Let $\G$ be the run \DAG of $\A$ on $w$ and let
$\q_i = \zug{S,\preceq,l,\lreqd,G,B}$ be the $i$-th macrostate in the run of $D^R(\A)$ on $w$:
\begin{compactenum}
\item\label{State_Sets_Match} $S = \set{q\mid \rzug{q,i} \in \G}$.
\item\label{Preceq_Match} For $q, r \in S$, it holds that $q \preceq r$ iff $\rzug{q,i} \preceq_i \rzug{r,i}$.
\item\label{Edges_Match} For $q \in S$ and $q' \in Q$, it holds that $q' \in \rhoprec(q,\sigma_i)$ iff $\zug{\rzug{q,i},\rzug{q',i\splus 1}} \in E'$.
\end{compactenum}
\end{lemma}
\begin{proof}
\item We proceed by induction over $i$, at each step proving \ref{State_Sets_Match} and
\ref{Preceq_Match} for $i+1$, and proving \ref{Edges_Match} for $i$.  As a base case, for $i=0$, we
have $S=Q^{in}$and $\preceq= Q^{in} \times Q^{in}$.  As $Q^{in} \cap F = \emptyset$, for every $u,
v$ on level $0$ of $\Gprime$ $h_u=0=h_v$ and $u \preceq_0 v$.  Inductively, assume that
\ref{State_Sets_Match} and \ref{Preceq_Match} holds for $\q_i= \zug{S,\preceq,l,\lreqd,G,B}$, and let
$\q_{i+1}= \zug{S',\preceq',l',\lreqd',G,B}$ be the $\sigma$-successor of $\q_i$. We show that 
\ref{Edges_Match} holds for $\q_i$, and \ref{State_Sets_Match} \ref{Preceq_Match} holds for $\q_{i+1}$.

\begin{compactenum}
\item As $S' = \rho(S,\sigma_i)$, by the inductive hypothesis and the definition of $W$ we have $S' =
\set{q'\mid \rzug{q',i\splus 1} \in \G}$. 

\item By definition, $q' \in  \rhoprec(q,\sigma_i)$ iff $q' \in \rho(q,\sigma_i)$ and for every $r \in
S$, if $q' \in \rho(r,\sigma_i)$ then $r \preceq q$. By the definition of $\G$ and the inductive
hypothesis, this holds iff $\zug{\rzug{q,i},\rzug{q',i\splus 1}} \in E$ and for every $\rzug{r,i}$,
if $\zug{\rzug{r,i},\rzug{q',i\splus 1}} \in E$, then $\rzug{r,i} \preceq_i \rzug{q,i}$. This it the
definition of $\zug{\rzug{q,i},\rzug{q',i\splus 1}} \in E'$, and thus \ref{Edges_Match} holds for $\q_i$.

\item For $q', r' \in S'$, let $q, r$ be such that $q' \in \rhoprec(q,\sigma_i)$ and $r' \in
\rhoprec(r,\sigma_i)$. By the inductive hypothesis, this implies
$\zug{\rzug{q,i},\rzug{q',i\splus 1}} \in E'$ and $\zug{\rzug{r,i},\rzug{r',i\splus 1}} \in E'$.  By
the definition of $\preceq'$, it holds that $q' \preceq' r'$ iff (a) $q \prec r$, (b) $q \approx r$
and $q' \in F$ iff $r' \in F$, or (c) $q \approx r$, $q' \not \in F$, and $r' \in F$. Recall that
$f$ is the function assigning $1$ to $F$-nodes, and $0$ to non-$F$-nodes. By the inductive
hypothesis, then, $q' \preceq r'$ iff (a) $\rzug{q,i} \prec_i \rzug{r,i}$, (b) $\rzug{q,i} \approx_i
\rzug{r,i}$ and $f(\rzug{q',i\splus 1})=f(\rzug{r',i\splus 1})$, or (c) $\rzug{q,i} \approx
\rzug{r,i}$, $f(\rzug{q',i\splus 1})=0$ and $f(\rzug{r',i\splus 1})=0$. By Lemma
\ref{Gprime_Captures_Profiles}, these are precisely the situations in which $\rzug{q',i\splus 1}
\preceq_{i+1} \rzug{r',i\splus 1}$.
\end{compactenum}
\end{proof}

Lemma \ref{Labels_Match} demonstrate the correlation the labeling $l'$ and the labeling $l_i$
of the tree of equivalence classes.  Lemma \ref{LR_Match} shows that, so defined, the preorder
$\lreqd$ describes the minimal cousin relation of Definition~\ref{Def:Min_Cousin}.  We
simultaneously prove Lemmas \ref{Labels_Match} and \ref{LR_Match} by induction. 

\begin{lemma}\label{Labels_Match}
Let $\G$ be the run \DAG of $\A$ on $w$ and $\q_i = \zug{S,\preceq,l,\lreqd,G,B}$ be the $i$th
macrostate in the run of $D^R(\A)$ on $w$.  For $q \in S$, it holds that $l(q) =
l_i([\rzug{q,i}])$.  
\end{lemma}

\begin{lemma}\label{LR_Match}
Let $\G$ be the run \DAG of $\A$ on $w$ and 
$\q_i = \zug{S,\preceq,l,\lreqd,G,B}$ be the $i$th macrostate in the run of $D^R(\A)$ on $w$. 
For $q, r \in S$ it holds that $q \lreqb r$ iff $[\rzug{q,i}] \lreq{i} [\rzug{r,i}]$
\end{lemma}
\begin{proof}
We prove these by induction over $i$. As a base case, for $i=0$, we have $S=Q^{in}$, $\lreqd= Q^{in}
\times Q^{in}$, and $l(q)=0$ for every $q \in S$.  By definition, the $0$th level of $\Gprime$ is
$\set{\rzug{q,0}\mid q \in Q^{\in}}$. As $Q^{in} \cap F = \emptyset$, for every $u, v$ on level $0$ of
$\Gprime$ $h_u=0=h_v$ and $u \preceq_0 v$.  Since there is only one equivalence class $U$, we have
$U \lreq{0} U$, and $l(U)=0$. 

Inductively, assume this holds for $\q_i= \zug{S,\preceq,l,\lreqd,G,B}$, and let $\q_{i+1}=
\zug{S',\preceq',l',\lreqd',G,B}$ be the $\sigma$-successor of $\q_i$. Note by Lemma
\ref{Sigma_Successors_Valid} that $l'$ gives unique labels to the equivalence classes of
$\preceq'$, and $\lreqd'$ is a partial order over the equivalence classes of $\preceq'$. 

\standout{Proof of Prop. \ref{Labels_Match}} For $q' \in S'$, we prove $l'(q') =
l_{i+1}([\rzug{q',i\splus 1}])$ as follows. First, by definition for $q \in S$ and $r' \in S'$, $r'
\in  \neph(q,\sigma)$ if there exists $r \in S$ so that $q \lreqb r$ and $r' \in
\rhoprec(r,\sigma_i)$. By Lemma \ref{Gprime_Match}.\ref{Edges_Match}, the inductive hypothesis, and the definition of $\TGP$, this
holds if there is a $W$ so that $[\rzug{r',i\splus 1}]$ is a child of $W$ and $[\rzug{q,i}] \lreq{i}
W$: the definition of $[\rzug{r',i\splus 1}] \in \lsf_i([\rzug{q,i}])$.  Thus $\lsf_i([\rzug{q,i}])=
\set{[\rzug{r',i\splus 1}] \mid r' \in \neph(q,\sigma)}$. By Lemma \ref{Gprime_Match}.\ref{Preceq_Match} this implies for every $r' \in
S'$, $\lpf_i([\rzug{r',i\splus 1}]) = \set{[\rzug{q,i}] \mid q \in \unc(r',\sigma)}$. We have two
cases. If $\unc(q',\sigma) \neq \emptyset$, then $\lpf_i([\rzug{q',i\splus 1}]) \neq \emptyset$, and
$l'(q') = l(q)$ for $q \in \unc(q',\sigma)$. By the inductive hypothesis and Lemma
\ref{Gprime_Match}.\ref{Preceq_Match} this implies $[\rzug{q,i}] =
\min_{\preceq_i}(\lpf_i([\rzug{q',i\splus 1}]))$ and $l_{i+1}([\rzug{q',i\splus
1}])=l_i([\rzug{q,i}])=l(q)$. Alternately, if $\unc(q',\sigma) = \emptyset$, then
$\lpf_i([\rzug{q',i\splus 1}]) = \emptyset$.  The inductive hypothesis implies that $\mathrm{FL}(l_i)$, the set of
unused labels in $l_i$, is identical to $\mathrm{FL}(l)$, the set of unused labels in $l$.  Thus the
$\zug{\preceq_{i+1},<}$-minjection from the classes on level $i\splus 1$ of $\TGP$ to
$\mathrm{FL}(l_i)$
corresponds to the  $\zug{\preceq',<}$-minjection from $S'$ to $\mathrm{FL}(l)$, and
$\fl(q')=\fl_{i+1}([\rzug{q,i\splus 1}])$. 

\standout{Proof of Prop. \ref{LR_Match}} There are two cases in which $q' \lreqbp r'$. First, if $q'
\approx' r'$, then by Lemma \ref{Gprime_Match}.\ref{Preceq_Match} $[\rzug{q',i\splus 1}] = [\rzug{r',i\splus 1}]$
and, as $\lreq{i+1}$ is reflexive, $[\rzug{q',i\splus 1}] \lreq{i+1} [\rzug{r',i\splus 1}]$.
Otherwise $q' \not\approx' r'$ and $q' \lreqbp r'$ iff there exists $r, q_2 \in S$ so that $q_2 \in
\unc(q')$, $r' \in \rhoprec(r,\sigma_i)$, and $q_2 \lreqb r$. By Lemma
\ref{Gprime_Match}.\ref{Preceq_Match} and \ref{Edges_Match} this entails $q' \lreqbp r'$ iff there
exists $U$ and $W$ so that $W$ is the parent of $[\rzug{r',i\splus 1}]$, $l_i(U) =
l_{i+1}([\rzug{q',i\splus 1}])$, and $U \lreq{i} W$. By Lemmas \ref{LMD_Successor} and
\ref{Rabin_Corresponds}, this is precisely the condition under which $[\rzug{q',i\splus 1}]
\lreq{i+1} [\rzug{r',i\splus 1}]$.
\end{proof}

Lemma \ref{Acceptance_Matches} shows that the presence of a label in $G$ corresponds to the
success success of a label in $l_i$.

\begin{lemma}\label{Acceptance_Matches}
Let $\G$ be the run \DAG of $\A$ on $w$ and $\q_i = \zug{S,\preceq,l,\lreqd,G,B}$ be the $i$th
macrostate in the run of $D^R(\A)$ on $w$.  For every label $m$, it holds that $m \in G$ iff $m$ is
good in $l_i$ and $m \in B$ iff $m$ is bad in $l_i$. 
\end{lemma}
\begin{proof}
Let $\q_{i} = \zug{S,\preceq,l,\lreqd,G,B}$ and $\q_{i+1} = \zug{S',\preceq',l',\lreqd',G,B}$. Recall that,
with respect to $i$, $R = \set{r \in S \mid l(r)=m}$ and $R' = \set{r' \in S' \mid l'(r')=m}$.
By Lemma \ref{Sigma_Successors_Valid}, we have that $R$ is an equivalence class under $\preceq$ and
$R'$ is an equivalence class under $\preceq'$.  By definition, $m$ dies in $l_i$ when $m$ is in the
range of $l_{i}$, but not in the range of $l_{i+1}$.  By Lemma \ref{Labels_Match}
this is true iff $R \neq \emptyset$, but $R' = \emptyset$: the definition of $m$ dies in
$\zug{\q_{i},\q_{i+1}}$. 

Similarly, $m$ succeeds in $l_{i+1}$ if there are classes $U$ on level $i-1$ and $U'$ on level $i$ so
that $l_{i}(U)=l_{i+1}(U')=m$, and $U'$ is not the non-$F$-child of $U$.  By Lemmas
\ref{Gprime_Match}.\ref{State_Sets_Match} and \ref{Labels_Match}, a $U$ and $U'$ exist so that
$l_{i}(U)=l_{i+1}(U')=m$, iff $U=\set{\rzug{r,i}\mid r \in R}$ and $U'=\set{\rzug{r',i+1}\mid r' \in
R'}$. This entails that $m$ succeeds in $l_{i+1}$ iff $R \neq \emptyset$ and $R' \neq \emptyset$, 
and either $U'$ is an $F$-class, or $U'$ is not a child of $U$.  If $U'$ is an $F$-class then $R'
\subseteq F$.  If $U'$ is not a child of $U$, then by the definition of $\TGP$ and
Lemma~\ref{Gprime_Match}.\ref{Edges_Match} there is no $r \in R$, $r' \in R'$ where $r' \in
\rhoprec(r,\sigma_i)$.  This entails $\set{\rhoprec(r,\sigma) \mid r \in R} \cap R' = \emptyset$. We
conclude that $m$ succeeds in $l_{i+1}$ iff $m$ succeeds in $\zug{\q_{i},\q_{i+1}}$.
\end{proof}

Finally, we must bound the number of preorders $\lreqd$ to bound the size of the automaton.

\begin{lemma}\label{LR_Partial_Order}
For a level $i$, the preorder $\lreq{i}$ is a tree order over the classes on level $i$ of
$\TGP$. 
\end{lemma}
\begin{proof}
Let $\U$ be the set of classes on level $i$ of $\TGP$. By definition, $\lreq{i}$ is a tree order if
for every $W \in \U$, the $\set{U \mid U \leq W}$ is totally ordered by $\lreq{i}$.  Consider two
classes $U_1 \lreq{i} W$ and $U_2 \lreq{i} W$. By definition, $W$ is a descendant of both
$\first(\gl(U_1))$ and $\first(\gl(U_2))$. Since $\TGP$ is a tree, one of $\first(\gl(U_1))$ or
$\first(\gl(U_2))$ is a descendant of the other.  Without loss of generality, assume $\first(\gl(U_1))$
is a descendant of $\first(\gl(U_2))$. Since $U_1$ is a descendant of $\first(\gl(U_1))$, it is a
descendant of $\first(\gl(U_2))$ too, and $U_2 \lreq{i} U_1$. 
\end{proof}

\thmConstructionCorrect*
\begin{proof}
That $L(D^R(\A)) = L(\A)$ follows from Theorem \ref{Rabin_Labeling_Succeeds} and
Lemma \ref{Acceptance_Matches}. To bound the number of macrostates
$\zug{S,\preceq,l,\lreqd,G,B}$,
we observe that the number of subsets $S$ and linear orders $\preceq$ is $n^{O(n)}$ \cite{Var80}.
The number of labelings is likewise $n^{O(n)}$. By Lemma \ref{LR_Partial_Order} and Lemma
\ref{LR_Match}, $\lreqd$ is a tree-order over the equivalence classes of $S$ under $\preceq$. By
Cayley's formula, the number of tree orders is bounded by $n^{n-2}$.  Thus the number of
macrostates is bounded by $n^{O(n)}$. 
\end{proof}

\section{Smaller Constructions}\label{App:Constructions}

We here present two variants of Definition \ref{Def:Rabin_Condition}, both of which only use a
variation of macrostates where labels are restricted to $\set{0,\ldots,\abs{Q}-1}$.  The first is a
\emph{Rabin-edge automaton}, in which the acceptance condition is a set
$\zug{G_0,B_0},\ldots,\zug{G_k,B_k}$ of pairs of sets of transitions: thus $G_j, B_j \subseteq Q^2$
for $0 \leq j \leq k$.  A run is accepting iff there exists $0 \leq j\leq k$ so that
$\zug{q_i,q_{i+1}} \in G_j$ for infinitely many $i$'s, while $\zug{q_i,q_{i+1}} \in B_j$ for only
finitely many $i$'s. The second is a \emph{parity-edge automaton}, the acceptance condition is a
parity function $\pri \colon Q^2 \to \set{0,\ldots,k}$, and a run is accepting if the smallest
element of $\set{j\mid j=\pri(q_i,q_{i+1}) \text{ for infinitely many $i$'s}}$ is even. Define the
set of \emph{tight macrostates} to be four-tuples $\zug{S,\preceq,l,\lreqd}$, where $S$, $\preceq$,
and $\lreqd$ are defined as for normal macrostates, and where $l \colon S \to \set{0,\ldots, \abs{Q}-1}$
is a tighter labeling. Let $\bQ^t$ be the set of tight macrostates.

\subsection{Tight Rabin Variant:}
Given a tight macrostate $\q \in \bQ^t$ and $\sigma \in \Sigma$, define the \emph{Rabin
$\sigma$-successor of \q} to be $\qp = \zug{S',\preceq',l',\lreqd'}$ where $S'$, $\preceq'$, and
$\lreqd'$ are defined as in Section \ref{Sect:Definition}, and $l'$ is defined as follows:
\begin{compactenum}
\item For $q \in S$, let $\neph(q,\sigma) = \min_{\preceq'}(\set{r' \mid \text{exists }r \in S,~q \lreqb r,~r' \in
\rhoprec(r,\sigma)})$, as in Section \ref{Sect:Definition}.
\item For $r' \in S'$, let $\unc(r',\sigma)=\min_\preceq(\set{q \mid r' \in \neph(q,\sigma)})$, as
in Section \ref{Sect:Definition}.
\item $\mathrm{FL}(l)= \set{m\mid m\text{ not in the range of $l$}}$ $\cup$ $\set{l(q) \mid\text{ for every $r' \in S'$}, q \not \in \unc(r',\sigma)}$. 
\item $\fl$ is the $\zug{\preceq',<}$-minjection from $\{r' \in S'\mid \unc(r',\sigma = \emptyset\}$ to
$\mathrm{FL}(l)$.
\item For $r' \in S'$, let $l'(r') =
\begin{cases} 
l(q),~q \in \unc(r',\sigma) & \text {if } \unc(r',\sigma) \neq \emptyset,\\
\fl(r') & \text{if }\unc(r',\sigma) = \emptyset.\\
\end{cases} $
\end{compactenum}
\bigskip

For $\sigma \in \Sigma$ and label $m \in \set{0, \ldots, \abs{Q}-1}$, given a tight macrostate $\q =
\zug{S,\preceq,l,\lreqd} \in \bQ^t$ and its Rabin $\sigma$-successor $\qp =
\zug{S',\preceq',l',\lreqd'}$ let $R = \zug{r \in S \mid l(r)=m}$ and $R' = \zug{r' \in S' \mid
l'(r')=m}$.  Say that $m$ \emph{Rabin-dies in $\zug{\q,\qp}$} when $R \neq \emptyset$ and $m \in
\mathrm{FL}(l)$.
Say that $m$ \emph{Rabin-succeeds in $\zug{\q,\qp}$} when it does not die in $\zug{\q,\qp}$, $R \neq
\emptyset$, $R' \neq \emptyset$, and either $R' \subseteq F$ or $\rhoprec(R,\sigma) \cap R' =
\emptyset$.

\begin{definition}\label{Def:Tight_Rabin_Condition}
Define the \DREW automata $D^T(\A)$ to be $\zug{\Sigma,\bQ^t, \bQ^{in},
{\bf \rho_Q},\alpha}$ where:
\begin{compactitem}
\item $\bQ^{in}$ is as defined in Definition \ref{Def:Rabin_Condition}
\item For $\q \in \bQ^t$ and $\sigma \in \Sigma$, ${\bf \rho_Q}(\q,\sigma)= \set{\qp}$ where $\qp$
is the Rabin $\sigma$-successor of $\q$.
\item $\alpha=  \zug{G_0,B_0},\ldots,\zug{G_{\abs{Q}-1},B_{2\abs{Q}-1}}$ where for a label $m \in
\set{0, \ldots, 2\abs{Q}}$:
\begin{compactitem}
\item $G_m = \set{\zug{\q,\qp} \mid m \text{ Rabin-succeeds in }\zug{\q,\qp}}$.
\item $B_m=\set{\zug{\q,\qp} \mid m \text{ Rabin-dies in }\zug{\q,\qp}}$
\end{compactitem}
\end{compactitem}
\end{definition}

\begin{theorem}\label{Tight_Construction_Correct}
For an NBW $\A$, $L(D^T(\A)) = L(\A)$.
\end{theorem}
\begin{proof}
For every word $w$, we show that the run $\q_0,\q_1,\ldots$ of $D^R(\A)$ on $w$ is accepting iff the
run $\q^p_0,\q^p_1,\ldots$ of $D^P(\A)$ on $w$ is accepting.  For convenience, let
$\q_i=\zug{S_i,\preceq_i,l_i,\lreq{i}}$. We first note that for every $i$, it holds that
$\q^p_i==\zug{S_i,\preceq_i,l_i^p,\lreq{i}}$: that that is to say $\q_i$ and $\q^p_i$ match on
$S_i$, $\preceq_i$, and $\lreq{i}$. For $S$ and $\preceq$, this is easy to see: the definitions of
$S'$ and $\preceq$' are identical in $\sigma$-successors and Rabin-$\sigma$-successors. For
$\lreqd$, this follows from the fact that $\lreqd'$ is defined solely with respect to $\rhoprec$
and $\unc$, which do not change from $\sigma$-successors to Rabin-$\sigma$-successors.  We pause to
note that, for every $i$ and $q \in S_i$, $q' \in S_{i+1}$, we have that $l_{i+1}(q')=l_i(q)$, iff
$q \in \unc(q',\sigma_i)$, which holds iff both $l^p_{i+1}(q')=l^p_i(q)$ and $l^p_i(q) \not \in \in
\mathrm{FL}({l^p_i})$.

In one direction, assume there is a label $m$ that dies in finitely many $\zug{\q_{i},\q_{i+1}}$,
and succeeds in infinitely many $\zug{\q_{i},\q_{i+1}}$.  We pause to note that, for every $i$ and
$q \in S_i$, $q' \in S_{i+1}$, we have that $l_{i+1}(q')=l_i(q)$, iff $q \in \unc(q',\sigma_i)$,
which holds iff both $l^p_{i+1}(q')=l^p_i(q)$ and $l^p_i(q) \not \in \in \mathrm{FL}({l^p_i})$.  Let $j$ be
the first index so that $m$ occurs in $\q_{j}$, but for every $k > j$, $m$ does not die in
$\zug{\q_{k},\q_{k+1}}$. Let $q \in S_j$ be such that $l_j(q)=m$, and let $m'=l^p_j(q)$.  For $k >
j$, define $R_k = \set{r \in S_k \mid l_k(r)=m}$, and $R^p_k= \set{r \in S_k \mid l^p_k(r)=m'}$.
Since $m$ does not die in $\zug{\q_{k},\q_{k+1}}$, $R_k$ and $R_{k+1}$ are both non-empty, and by
our above observations $m'$ does not Rabin-die in $\zug{\q^p_{k},\q^p_{k+1}}$. Further, $R^p_k=R_k$,
and $R^p_{k+1}=R_{k+1}$.  This implies that if $m$ succeeds in $\zug{\q_k,\q_{k+1}}$, then $m'$
Rabin-succeeds $\zug{\q^p_k,\q^p_{k+1}}$. Thus $m'$ Rabin-dies in finitely many
$\zug{\q^p_{i},\q^p_{i+1}}$, and Rabin-succeeds in infinitely many $\zug{\q^p_{i},\q^p_{i+1}}$, and
$D^T(\A)$ accepts $w$.

In the other direction if $D^T(\A)$ accepts $w$, this implies is a label $m$ that Rabin-dies in finitely many
$\zug{\q^p_{i},\q^p_{i+1}}$, and Rabin-succeeds in infinitely many $\zug{\q^p_{i},\q^p_{i+1}}$.  Let
$j$ be the first index so that $m$ occurs in $\q^p_{j}$, but for every $k > j$, $m$ does not
Rabin-die in $\zug{\q^p_{k},\q^p_{k+1}}$. Let $q \in S_j$ be such that $l^p_j(q)=m$, and let
$m'=l_j(q)$.  For $k > j$, define $R^p_k= \set{r \in S_k \mid l^p_k(r)=m}$, and $R_k = \set{r \in
S_k \mid l_k(r)=m'}$, Since $m$ does not Rabin-die in $\zug{\q^p_{k},\q^p_{k+1}}$, $m \not \in
F_{l^p_k}$ and $R_k$, $R_{k+1}$ are both non-empty. By our above observations $m'$ does not die in
$\zug{\q_{k},\q_{k+1}}$. Further, $R^p_k=R_k$, and $R^p_{k+1}=R_{k+1}$.  This implies that if $m$
Rabin-succeeds in $\zug{\q^p_k,\q^p_{k+1}}$, then $m'$ succeeds $\zug{\q_k,\q_{k+1}}$. Thus $m'$
dies in finitely many $\zug{\q_{i},\q_{i+1}}$, and succeeds in infinitely many
$\zug{\q_{i},\q_{i+1}}$, and $D^R(\A)$ accepts $w$.
\end{proof}

\subsection{Parity Variant}
The parity variation simply shifts labels down, instead of giving arbitrary free labels to new nodes. This
means labels in the automaton are no longer consistent with with the labels $l_i$ over \TGP. To
simplify this, we use an intermediate labeling that keeps labels consistent between two levels, but
can use the labels $\set{\abs{Q},\ldots 2\abs{Q}}$.  Given a tight macrostate $\q \in \bQ^t$ and
$\sigma \in \Sigma$, define the \emph{parity $\sigma$-successor of \q} to be $\qp =
\zug{S',\preceq',l',\lreqd'}$ where $S'$, $\preceq'$, and $\lreqd'$ are defined as in Section
\ref{Sect:Definition}, and $l'$ is defined as follows:
\begin{compactenum}
\item For $q \in S$, let $\neph(q,\sigma) = \min_{\preceq'}(\set{r' \mid \text{exists }r \in S,~q \lreqb r,~r' \in
\rhoprec(r,\sigma)})$
\item For $r' \in S'$, let $\unc(r',\sigma)=\min_\preceq(\set{q \mid r' \in \neph(q,\sigma)})$
\item $\fl$ is the $\zug{\preceq',<}$-minjection from $\{r' \in S'\mid \unc(r',\sigma = \emptyset\}$ to
$\set{\abs{Q},\ldots,2\abs{Q}}$
\item For $r' \in S'$, define the intermediate labeling\\ $l^{int}(r') =
\begin{cases} 
l(q),~q \in \unc(r',\sigma) & \text {if } \unc(r',\sigma) \neq \emptyset,\\
\fl(r') & \text{if }\unc(r',\sigma) = \emptyset.\\
\end{cases} $
\item For $r' \in S'$, define the final labeling $l'(r') = \abs{\set{l^{int}(q') \mid l^{int}(q') <
l^{int}(r')}}$
\end{compactenum}
\bigskip

For $\sigma \in \Sigma$ and label $m \in \set{0, \ldots, \abs{Q}-1}$, given a tight macrostate $\q =
\zug{S,\preceq,l,\lreqd} \in \bQ^t$ and its parity $\sigma$-successor $\qp =
\zug{S',\preceq',l',\lreqd'}$ let $l^{int}$ be the intermediate labeling defined above.  Let $R =
\zug{r \in S \mid l(r)=m}$ and $R' = \zug{r' \in S' \mid l^{int}(r')=m}$.  Note that $R'$ is defined
with respect to the intermediate labeling.  Say that a label $m$ \emph{parity-dies in
$\zug{\q,\qp}$} if $m \in R$, but $m \not\in R'$.  Say that $m$ \emph{parity-succeeds in
$\zug{\q,\qp}$} when $R \neq \emptyset$, $R' \neq \emptyset$, and either $R' \subseteq F$ or
$\rhoprec(R,\sigma) \cap R' = \emptyset$. Define the priority function $\pri \colon \bQ^t \times \bQ^t \to
\set{1,\ldots,2\abs{Q}}$ so that $\pri(\zug{\q\qp})$ is $\min(\set{2m+2\mid\text{$m$
parity-succeeds in $\zug{\q,\qp}$}} \cup \set{2m+1\mid\text{$m$ parity-dies in $\zug{\q,\qp}$}})$.

\begin{definition}\label{Def:Parity_Condition}
Define the \DPEW automata $D^P(\A)$ to be $\zug{\Sigma,\bQ^t, \bQ^{in}, {\bf \rho_Q},\pri}$ where:
\begin{compactitem}
\item $\bQ^{in}$ is as defined Definition \ref{Def:Rabin_Condition}
\item For $\q \in \bQ$ and $\sigma \in \Sigma$, ${\bf \rho_Q}(\q,\sigma)= \set{\qp}$ where $\qp$
is the parity $\sigma$-successor of $\q$.
\end{compactitem}
\end{definition}

\begin{theorem}\label{Parity_Construction_Correct}
For an NBW $\A$, $L(D^P(\A)) = L(\A)$.
\end{theorem}
\begin{proof}
As above, for every word $w$, we show that the run $\q_0,\q_1,\ldots$ of $D^R(\A)$ on $w$ is
accepting iff the run $\q^p_0,\q^p_1,\ldots$ of $D^P(\A)$ on $w$ is accepting.  For convenience, let
$\q_i=\zug{S_i,\preceq_i,l_i,\lreq{i}}$. Again, it holds that for every $i$
$\q^p_i==\zug{S_i,\preceq_i,l_i^p,\lreq{i}}$: $\q_i$ and $\q^p_i$ match on $S_i$, $\preceq_i$, and
$\lreq{i}$. For every $i$, let $l^{int}_i$ be the intermediate labeling defined above.  However, it
is no longer that case that the labels of a branch will be consistent from $\q^p_i$ to $\q^p_{i+1}$.
Instead, we must look for consistency in the intermediate labeling.  for every $i$ and $q \in S_i$,
$q' \in S_{i+1}$, we have that $l_{i+1}(q')=l_i(q)$ iff $l^{int}_{i+1}(q')=l^p_i(q)$. If
$l^p_{i+1}(q') \neq l^{int}_{i+1}(q')$, this implies there was a label $n < l^p_i(q)$ that occurs in
the range of $l^p_i$, but not in the range of $l^{int}_{i+1}$.

In one direction, assume there is a label $m$ that dies in finitely many $\zug{\q_{i},\q_{i+1}}$,
and succeeds in infinitely many $\zug{\q_{i},\q_{i+1}}$.  Let $j$ be the first index so that $m$
occurs in $\q_{j}$, but for every $k > j$, $m$ does not die in $\zug{\q_{k},\q_{k+1}}$. For every
$j' > j$, let $q_{j}' \in \S_{j'}$ be such that $l_{j'}(q_{j'})=m$.  Note that the values of
$l^p_{j'}(q_{j'})$ can only decrease: new labels are only introduced above $\abs{Q}$, and
$l^p_j(q_j) < \abs{Q}$. Thus at some point the labels of $q_{j'}$ cease decreasing, and reach a
stable point. Let $k$ be this point, and let $m' = l^p_k(q_k)$.  For a level $k' > k$, define
$R_{k'} = \set{r \in S_{k'} \mid l_{k'}(r)=m}$, and $R^p_{k'}= \set{r \in S_{k'} \mid
l^{int}_{k'}(r)=m'}$.  Since the labels of $q_{k'}$ have stopped decreasing, we have that
$R^p_{k'}=R_{k'}$. For every $k' > k$, it holds that $m'$ does not parity-die  in
$\zug{\q^p_{k'},\q^p_{k'+1}}$. Further, every label $n < m'$ must occur on every level $k' > k$:
otherwise $l^p_{k'}(q_{k'})$ would not equal $l^{int}_{k'}(q_{k'})$. Thus for every $k' > k$, there
is no label $n < m'$ that parity-dies in $\zug{\q^p_{k'},\q^p_{k'+1}}$.  Therefore
$\pri(\q^p_{k'},\q^p_{k'+1}) \geq 2m'+1$. Now consider a level $k' > k$ where $m$ succeeds in
$\zug{\q_{k'},\q_{k'+1}}$. By the note above, $R^p_{k'}=R_{k'}$, and $m'$ parity-succeeds in
$\zug{\q^p_{k'},\q^p_{k'+1}}$, and $\pri(\q^p_{k'},\q^p_{k'+1}) = 2m'+2$. We have thus shown that
the smallest priority occurring infinitely often in $2m'+2$, and thus $w$ is accepted by $D^P(\A)$.

In the other direction if $D^P(\A)$ accepts $w$, this implies is a label $m$ and level $j$ so that 
for every $k > j$, it holds $\pri(\q^p_{k},\q^p_{k+1}) \geq 2m+2$, and for infinitely many $k > j$
it holds $\pri(\q^p_{k},\q^p_{k+1}) = 2m+2$. As noted above, this implies for every $k > j$ and
$n \leq m$, $n$ does not parity-die in $\zug{\q^p_{k},\q^p_{k+1}}$, and for infinitely many $k > j$,
$m$ parity-succeeds in $\zug{\q^p_{k},\q^p_{k+1}}$.
Thus we conclude that for
every $k > j$ and $q \in S_{k}$, $l^p_{k}(q)=m$ iff $l^{int}_{k}(q)=m$.  Let $q \in S_j$ be such
that $l^p_j(q)=m$, and let $m' = l_j(q)$. For every $k > j$, let $R^p_{k}= \set{r \in S_k \mid
l^p_k(r)=m}$, and let $R_{k} = \set{r \in S_{k} \mid l_{k}(r)=m'}$. Again, we have that 
$R^p_{k} =R_{k}$, thus for every $k > j$, $m'$ does not die in $\zug{\q_{k},\q_{k+1}}$, and 
and for infinitely many $k > j$ we have $m'$ succeeds in  $\zug{\q_{k},\q_{k+1}}$. Thus $w$ is
accepted by and $D^R(\A)$.
\end{proof}

\end{document}